\documentclass[11pt]{article}

\usepackage{fullpage}

\usepackage[ruled]{algorithm2e}

\SetAlFnt{\small}
\SetAlCapFnt{\small}
\SetAlCapNameFnt{\small}
\SetAlCapHSkip{0pt}
\IncMargin{-\parindent}

\usepackage{latexsym}
\usepackage{amsmath}
\usepackage{amstext}
\usepackage{amsfonts}
\usepackage{amsopn}

\usepackage{ifthen}

\usepackage{enumitem}

\newtheorem{theorem}{Theorem}[section]
\newtheorem{lemma}[theorem]{Lemma}
\newtheorem{corollary}[theorem]{Corollary}
\newtheorem{proposition}[theorem]{Proposition}
\newtheorem{definition}{Definition}[section]

\newenvironment{ProofDummyEnv}{}{}

\newenvironment{proof}[1][none]{\begin{proofby}[#1]{}}{\end{proofby}}

\newenvironment{proofby}[2][none]{\par\noindent{\bf Proof:} #2
\renewenvironment{ProofDummyEnv}{\begin{#1}}{\end{#1}}%
\begin{ProofDummyEnv}}%
{\QED\end{ProofDummyEnv}}

\newcommand{\bg}[1]{\medskip\noindent{\bf #1}}
\newenvironment{proofof}[1]{\bg{Proof of #1 : }}{\QED}

\newcommand{\reals}{\mathbb R}

\newcommand{\QED}{\nopagebreak\hfill $\Box$}

\newcommand{\prob}[2][]{\text{\bf Pr}\ifthenelse{\not\equal{}{#1}}{_{#1}}{}\!\left[#2\right]}
\newcommand{\expect}[2][]{\text{\bf E}\ifthenelse{\not\equal{}{#1}}{_{#1}}{}\!\left[#2\right]}

\DeclareMathOperator{\argmax}{argmax}

\def\eps{\epsilon}

\def\om{\omega}
\def\Om{\Omega}

\def\to{\rightarrow}

\def\calm{{\cal M}}

\def\calp{{\cal P}}

\def\E{{\bf E}}

\newtheorem{thm}{Theorem}[section]
\newtheorem{cor}[thm]{Corollary}
\newtheorem{lem}[thm]{Lemma}
\newtheorem{prop}[thm]{Proposition}

\newtheorem{defn}[thm]{Definition}

\newtheorem{example}[thm]{Example}

\newcommand{\mon}{\mbox{component-wise monotone}}
\newcommand{\sli}{\text{monotone strongly loser independent allocation rule}}
\newcommand{\msli}{\text{monotone strongly loser independent allocation rule}}
\newcommand{\alg}{{\cal A}}
\newcommand{\algi}[1][i]{{{\cal A}_{#1}}}
\newcommand{\mech}{{\cal M}}

\newcommand{\crit}{\theta}

\newcommand{\criti}[1][i]{{\crit_{#1}}}

\newcommand{\bid}{b}
\newcommand{\bids}{{\mathbf \bid}}
\newcommand{\bidsmi}{{\mathbf \bid}_{-i}}
\newcommand{\bidi}[1][i]{{\bid_{#1}}}

\newcommand{\val}{v}
\newcommand{\vals}{{\mathbf \val}}
\newcommand{\valsmi}{\vals_{-i}}
\newcommand{\vali}[1][i]{{\val_{#1}}}

\newcommand{\rank}{r}

\newcommand{\strat}{s}
\newcommand{\strats}{{\mathbf \strat}}
\newcommand{\stratsmi}{\strats_{-i}}
\newcommand{\strati}[1][i]{{\strat_{#1}}}

\newcommand{\type}{v}
\newcommand{\types}{{\mathbf \type}}
\newcommand{\typesmi}{\types_{-i}}
\newcommand{\typei}[1][i]{{\type_{#1}}}

\newcommand{\decl}{d}
\newcommand{\decls}{{\mathbf \decl}}
\newcommand{\declsmi}[1][i]{\decls_{-{#1}}}
\newcommand{\decli}[1][i]{{\decl_{#1}}}

\newcommand{\util}{u}

\newcommand{\utili}[1][i]{{\util_{#1}}}

\newcommand{\tdist}{F}
\newcommand{\tdists}{{\mathbf \tdist}}
\newcommand{\tdistsmi}{\tdists_{-i}}
\newcommand{\tdisti}[1][i]{{\tdist_{#1}}}

\newcommand{\sdist}{\om}
\newcommand{\sdists}{{\mathbf \sdist}}
\newcommand{\sdistsmi}{\sdists_{-i}}
\newcommand{\sdisti}[1][i]{{\sdist_{#1}}}

\newcommand{\dist}{F}
\newcommand{\dists}{{\mathbf \dist}}
\newcommand{\distsmi}{\dists_{-i}}

\newcommand{\price}{p}

\newcommand{\pricei}[1][i]{{\price_{#1}}}

\newcommand{\alloc}{x}
\newcommand{\allocs}{{\mathbf \alloc}}
\newcommand{\allocsmi}{\allocs_{-i}}
\newcommand{\alloci}[1][i]{{\alloc_{#1}}}

\newcommand{\balloc}{x}
\newcommand{\ballocs}{{\mathbf \balloc}}

\newcommand{\balloci}[1][i]{{\balloc_{#1}}}

\newcommand{\comment}[1]{}

\author{Brendan Lucier\thanks{Microsoft Research. \texttt{brlucier@microsoft.com}}
\and 
Allan Borodin\thanks{Department of Computer Science, University of Toronto.  \texttt{bor@cs.toronto.edu}}}
\title{Price of Anarchy for Greedy Auctions}

\begin{document}

\maketitle

\begin{abstract}
We consider auctions in which greedy algorithms, paired with first-price or critical-price payment rules, are used to resolve multi-parameter combinatorial allocation problems. We study the price of anarchy for social welfare in such auctions.  We show for a variety of equilibrium concepts, including Bayes-Nash equilibrium and correlated equilibrium, the resulting price of anarchy bound is close to the approximation factor of the underlying greedy algorithm.
\end{abstract}







\section{Introduction}
\label{sec.intro}

 
%
The field of algorithmic mechanism design studies 
systems that depend upon interaction with participants whose behaviour is motivated by their own goals, rather than those of a designer.  Relevant solutions must therefore merge the computational considerations of computer science with the game-theoretic insights of economics.  
The focus of this paper is the multi-parameter domain of combinatorial allocation problems when the goal is to assign $m$ objects to $n$ agents in order to maximize 
the social welfare, subject to arbitrary downward-closed feasibility constraints.  This class includes all combinatorial auction problems that allow single-minded declarations including multi-unit combinatorial auctions, unsplittable flow problems, and many others.

\comment{ 
A canonical problem within this field is to solve an optimization problem in a setting where the input is controlled by selfish agents.  These agents may strategically misreport values to manipulate an algorithm's outcome for their own benefit, rather than that of the overall good. In an auction, for example,  the participants' values for the good(s) being sold is private information unknown to the auctioneer.  Thus, when an auctioneer implements an auction, she must be sensitive to the fact that participants will behave in the interest of obtaining valuable goods at low costs, rather than behaving as the auction designer wishes.  
} 


For the goal of optimizing social welfare, the celebrated Vickrey-Clarke-Groves (VCG) mechanism addresses game-theoretic issues in a strong sense.  In the absense of collusion, it induces full cooperation (ie. truthtelling) as a dominant strategy.  However, the VCG mechanism requires that the underlying welfare-maximization problem be solved exactly.  For all but the simplest settings, this optimality requirement is undesirable: exact maximization may be computationally intractible, it may require an unrealistic amount of communication from the buyers, and the resulting winner determination rules may be difficult to explain to a typical participant. 
%
%
One way to bypass these complexity issues is to design
new, specially-tailored mechanisms for specific assignment problems.  Indeed, there has been significant progress in designing dominant strategy incentive compatible (DSIC) alternatives to the VCG mechanism.  While this venture has been largely successful in settings where agent preferences are single-dimensional \cite{AT-01,BKV-05,LOS-99,MN-08}, 
general settings have proven more difficult.  It has been shown that the approximation ratios achieveable by DSIC mechanisms and their non-incentive compatible counterparts exhibit a large asymptotic gap for some problems \cite{PSS-08,DOB-11,DughmiV11,DV-12}.

Alternatively, one might study classes of ``natural" allocation algorithms, that appear intuitive as auction allocation rules, with the hope that they have desireable incentive properties when implemented as mechanisms.
As it turns out, for many combinatorial allocation problems, conceptually simple deterministic algorithms (e.g. greedy algorithms) meet or approach the best-known approximation factors subject to computational constraints \cite{LOS-99,MN-08,BKV-05,BB-04}.
These natural methods tend to be computationally efficient and easy for bidders to understand, which are desirable properties in auctions.  Unfortunately, such algorithms are not, in general, DSIC \cite{LOS-99,BL-10b}.
Rather than abandoning these methods in favor of other, potentially more complex, mechanisms, we are pursuing an alternative approach.
Namely, rather than striving for dominant strategy truthfulness, 
it may be acceptable for a system to admit strategic manipulation, so long as the designer's objectives are met after such manipulation occurs.  
To this end, we explore the performance of mechanisms at equilibria of bidder behavior, given an appropriate model of beliefs.
Broadly speaking, our motivating question is:
When can an algorithm be implemented as a mechanism that achieves high social welfare at {\it every} equilibrium?\footnote{Dominant strategy truthfulness of an approximation mechanism is conceptually stronger as a solution concept than that of a mechanism that approximates the optimal social welfare at every equilibrium.  However, as noted elsewhere  \cite{CKS-08}, Bayesian Nash equilibrium is not, strictly speaking,
a relaxation of dominant strategy truthfulness.
There exist truthful mechanisms whose approximation ratios are not preserved at all Nash equilibria, such as the famous Vickrey auction.}
And how robust are the resulting mechanisms to 
variations of the equilibrium concept?


We demonstrate that for combinatorial allocation problems, any ``greedy-like'' approximation algorithm can be converted into a mechanism that achieves nearly the same approximation factor 
at every equilibrium of bidder behaviour.  Our analysis is very general, and applies to a range of different equilibrium concepts, including pure and mixed Nash equilibria, Bayes-Nash equilibria, and Bayesian equilibria with correlated types. 
We are thus able to decouple computational issues from incentives issues for this class of algorithms, as one can design a greedy algorithm without considering its economic implications, and then apply a straightforward pricing scheme in order to achieve good performance at equilibrium.

Performance of games at equilibrium has been studied extensively in the algorithmic game theory literature as the \emph{price of anarchy} (\emph{POA}) of a given game
\footnote{For the purpose of this paper, we shall not consider cost minimization problems. We note that the price of anarchy concept was introduced in terms
of cost minimization games but to the best of our knowledge the only price of anarchy results for mechanism induced games apply to maximization problems.
}: the ratio between the optimal outcome and the worst-case outcome at any equilibrium \cite{Papa-01}.
Put into these terms, our goal is to convert an algorithm with approximation factor $c \geq 1$ into a mechanism whose price of anarchy is not much larger than $c$.


This paper is a synthesis and revision of results in 
original conference version of this work \cite{BL-10}
and 
results in the first author's thesis \cite{L-thesis11}. The 
paper is organized as follows. The remainder of this section
outlines our results and relates our work to recent papers in this area. Section \ref{sec.prelim} defines the necessary concepts and applications for
our results. 
Section \ref{sec:strong-loser-indep} introduces the concept of strongly loser-independence (generalizing
the loser-independence concept from \cite{CG-09}) which becomes 
the key property 
of greedy algorithms that we will exploit. Sections \ref{sec:mixedBNE} and
\ref{sec:ex-post} analyze 
(respectively) price of anarchy results for first-price and critical-price 
mechanisms.  
Section \ref{sec:conclusions} concludes with some open problems.

\subsection{Our Results}

The basic question of algorithmic mechanism
design is this: when can computationally efficient algorithms be converted into 
mechanisms that preserve approximation bounds with respect to truthfulness or 
POA results?  
We address the price of
anarchy considerations with respect to social welfare maximization for a broad 
class of allocation problems. 
In the full information and Bayesian setting, 
we study   
the price of anarchy for first price and critical price mechanisms 
derived from greedy algorithms. Roughly speaking (and in contrast to
results regarding approximation and truthful mechanisms) we are 
able to show that 
there is often little or no loss from the approximation ratio of
a greedy algorithm to the corresponding mechanism price of anarchy.

We 
consider one-shot auctions, in which the allocation problem is resolved once.
Following Christodoulou et al. \cite{CKS-08}, we focus our attention 
on the standard (in economics) \emph{incomplete information setting}, where the appropriate equilibrium concept is Bayes-Nash equilibrium. 
That is, we assume that agents' preferences are private, but drawn independently from commonly-known prior distributions, and that players apply strategies at equilibrium given this partial knowledge.
We pose the question:  can a given black-box approximation algorithm be converted into a mechanism that approximately preserves its approximation ratio at every Bayes-Nash Equilibrium?
We show that for a broad class of greedy algorithms, the answer is \emph{yes}. 

\medskip
\textbf{Theorem (informal):} Suppose $\alg$ is a greedy $c$-approximate allocation rule for a combinatorial allocation problem.  Then the auction that uses $\alg$ to choose allocations, and uses a pay-your-bid payment scheme, has a Bayes-Nash Price of Anarchy of at most $c + O(c^2 / e^c)$.
\medskip

We also show that the small (and exponentially decreasing) loss in our price of anarchy bound is necessary, by giving an example (for every $c \geq 2$) where the resulting price of anarchy is at least $c + \Omega(\frac{c}{e^{4c}})$.

We note that the mechanisms we consider are all \emph{prior-free}.  Thus, as in the full-information case, while we assume the existence of type distributions in order to model rational agent behaviour, our mechanism need not be aware of these distributions.
In the special case that each player's type distribution is a point mass, Bayes-Nash equilibrium reduces to standard Nash equilibrium.  Our mechanisms therefore also preserve approximation ratios at every (mixed or pure) Nash equilibrium of the full information game.  Our analysis also extends to the more general class of coarse correlated equilibria.  For the case of pure Nash equilibrium, our price of anarchy bound improves to $c$.  

As is standard, our bounds on the Bayesian price of anarchy will assume that agent types are distributed independently.  However, we show that a weaker bound of $O(c)$ holds when agent types are drawn from an arbitrary distribution over the space of all type profiles.  This result applies to greedy algorithms that are non-adaptive, as described in Section \ref{sec:greedy}.  Thus, even if agent types are arbitrarily correlated, our mechanisms yield performance at equilibrium asymptotically matching that of the underlying allocation algorithm.

A similar bound also applies to mechanisms that use the critical-price payment scheme, which is a natural extension of second-price payments in single-item auctions.  Such a payment scheme charges each bidder the minimum bid at which he would have maintained his allocation.  These bounds require a standard no-overbidding assumption, which is that agents are avoid bidding more than their value for any given subset of items.

\medskip
\textbf{Theorem (informal):} Suppose $\alg$ is a greedy $c$-approximate allocation rule for a combinatorial allocation problem.  Then, under the assumption that agents do not overbid, the auction that uses $\alg$ to choose allocations, and uses a critical-price payment scheme, has a Bayes-Nash Price of Anarchy of at most $c + 1$.
\medskip

We also show that the extra $+1$ term is essentially necessary for large $c$, by giving an example for every $c \geq 2$ where the resulting price of anarchy is at least\footnote{More specifically, for any $\gamma > 1$ there is an example in which a greedy algorithm has approximation factor $\gamma + \frac{1}{\gamma}$ and the resulting price of anarchy is at least $\gamma+1$.} $c + 1 - O(\frac{1}{c})$. 
As with the first-price results, our bounds extend to coarse correlated equilibria, and a bound of $O(c)$ holds if agent valuations can be correlated.  Furthermore, we show that a slight modification to the mechanism allows us to replace the no-overbidding assumption with the (conceptually weaker) assumption that bidders avoid weakly dominated strategies. 


\subsection{Related Work}

The seminal paper in algorithmic game theory and more specifically 
algorithmic mechanism design is that of Nisan and Ronen \cite{NR-99}. The basic issue introduced in \cite{NR-99} is to reconcile the competing demands for revenue and social welfare optimization with the need for computational efficiency in the context of self interested (i.e. selfish) agents. 
The two most studied solution concepts in algorithmic game theory are truthfulness (i.e. incentive compatability) and behavior at (all) equilbria (i.e. the price of anarchy (POA) concept). Initial POA results for games were first introduced to algorithmic game theory in the seminal papers by Papdimitriou \cite{Papa-01} and
Roughgarden and Tardos \cite{RT-00}.  
 Christodoulou et al. \cite{CKS-08} initiated the study of the price of the price of anarchy in the Bayesian setting. 
Whereas the emphasis of
algorithmic mechanism design has been to consider the approximations achieveable by truthful mechanisms, to the best of our knowledge, our conference paper 
\cite{BL-10} was the first to
consider this constructive aspect of mechanism design and price
of anarchy.  

Since the initial conference version of this work 
there has been significant progress on the understanding of the price of anarchy of mechanisms in various auction settings.  Some  examples include the Generalized Second Price auction for sponsored search ads \cite{CKKKLPTardos12}, simultaneous single-item auctions \cite{CKS-08,HKMN11,BR-11,FFGL13}, and multi-unit auctions \cite{MT12,dKMST13}. A framework unifying much of this work was proposed by Syrgkanis and Tardos \cite{SyrgkanisT13}.

Chekuri and Gamzu \cite{CG-09} defined ``loser-independent algorithms'', and in the  
conference version of our paper \cite{BL-10} we argued that the basic property of greedy algorithms that we were
exploiting was a multi-parameter version of loser-independence. 
In the first author's thesis \cite{L-thesis11}, a strengthening of loser-independence, called strong 
loser-independence, was introduced to
simplify the proofs and that will be the basic property of greedy algorithms we will use in this paper. 
Loser-independence is conceptually related to the concept of \emph{smoothness}, which was introduced by
Roughgarden \cite{ROU-09} as a general way to derive price of anarchy results
for one shot and repeated games (without reference to mechanisms 
that derive games).  Loser-independence has been shown to be different 
from this original notion of smoothness \cite{L-thesis11}. 
However, alternative notions of  
smoothness defined by Lucier and Paes Leme \cite{PLucier11} 
and Syrgkanis and Tardos \cite{SyrgkanisT13} can also be used to 
derive results similar to our results.
In particular, Syrgkanis and   
Tardos use their smoothness condition to derive many price of anarchy 
results for allocation mechanisms, including those derived from greedy 
$c$-approximation algorithms. Their result for the (non correlated) 
mixed Bayesian and coarse correlated equilbrium improved
upon our conference results: as in our current paper, they show
that the resulting price of anarchy approaches $c$ with a term
exponentially decreasing in $c$. In particular, they show that 
the price of anarchy is never worse than $c + .58$.
As we will show in Section~\ref{sec:applying-loser-indep}, our application of strong 
loser-independence can be interpreted as a proof of smoothness. 

%

\section{Preliminaries}
\label{sec.prelim}

\subsection{Feasible Allocation Problems}

We consider a setting in which there are $n$ agents and a set $M$ of $m$ objects.  An \emph{allocation} to agent $i$ is a subset
$\alloci \subseteq M$.  A \emph{valuation function} $\val : 2^M \to \reals$ assigns a value to each allocation.  We assume that valuation functions are monotone, meaning $\val(S) \leq \val(T)$ for all $S \subseteq T \subseteq M$, and normalized so that $\val(\emptyset) = 0$.  
A valuation function $\val$ is \emph{single-minded} if there exists a set $S \subseteq M$ and a value $y \geq 0$ such that for all $T \subseteq M, \val(T) = y$ if $S \subseteq T$ and 0 
otherwise.  
A \emph{valuation profile} $\vals$ is a vector of $n$ valuation functions, one for each agent.  In general we will use boldface to represent vectors, subscript $i$ to denote the $i$th component, and subscript $-i$ to denote all components except $i$, so that $\vals = (\vali,\valsmi)$.  An \emph{allocation profile} $\ballocs$ is a vector of $n$ allocations. The goal in our social welfare maximization problems is to choose an allocation for each agent in order to maximize the sum of agent values. 

 A \emph{combinatorial allocation problem} is defined by a set of \emph{feasible allocations}, which is the set of permitted allocation profiles. We further assume in combinatorial allocation problems that this feasibility constraint is \emph{separable}, meaning that if $\allocs$ is feasible then $(\emptyset, \allocsmi)$ is also feasible\footnote{We note that the combinatorial public projects problem 
(CPPP) \cite{PSS-08} is not separable.} for all $i$.  Note that separability is a weaker assumption than the standard downward-closure property of packing problems, which would stipulate that if $\allocs$ is feasible then $(y_i,\allocsmi)$ is also feasible for all $y_i \subseteq \alloci$.
An \emph{allocation rule} $\alg$ assigns to each valuation profile $\vals$ a feasible outcome $\alg(\vals)$; we write $\alg_i(\vals)$ 
for the allocation to agent $i$. 
An allocation rule is {\it \mon} if it satisfies the following property for
every agent $i$: 
$$\mbox{If } \vali(S) < {\tilde \vali}(S), \ 
\vali(T) = {\tilde \vali}(T) \ \forall T \neq S,
\mbox{ and } 
\alg_i(\vali,\valsmi) = S, \mbox{ then } \alg_i({\tilde \vali},\valsmi) = S  $$
We will tend to write $\alg$ for both an allocation rule and an algorithm that implements it.
We will sometimes abuse notation and use $\allocs$ for an allocation rule, rather than a specific 
allocation. 

Each agent $i \in [n]$ has a private valuation function $\typei$, his \emph{type}, which defines the value he attributes to each allocation.
The \emph{social welfare} obtained by allocation profile $\ballocs$, given type profile $\types$, is $SW(\ballocs,\types) = \sum_i \typei(\balloci)$.  We write $SW_{opt}(\types)$ for $\max_{\ballocs}\{SW(\ballocs,\types)\}$
and say that algorithm $\alg$ is a $c$ approximation algorithm 
\footnote{Our convention will be to have approximation ratios $c \geq 1$.} 
if $SW(\alg(\types),\types) \geq \frac{1}{c}SW_{opt}(\types)$ for all 
$\types$.

A type profile $\types$ and an allocation rule $\alg$ for a combinatorial allocation problem define {\em critical values}, $\criti(S,\valsmi)$, 
for any agent $i$ 
and set 
$S \subseteq  M$. The value $\criti(S,\valsmi)$ is the 
minimum amount that agent 
$i$ could bid on set $S$ and still win $S$, assuming the other agents 
have profile $\valsmi$. That is,
$\criti(S,\valsmi) = \inf\{z$ : $x_i(z, \valsmi) = S\}$.
We note that this definition of critical values holds even if it is not the case that increasing one's value for a set necessarily increases the probability of obtaining that set. However, most of the mechanisms we consider in this work do satisfy this monotonicity property, which motivates the terminology of a “critical” price.

\subsection{Mechanisms}

 
A \emph{direct revelation mechanism} $\calm(\alg,P)$ is composed of an allocation rule $\alg$ and a payment rule $P$ that assigns a vector of $n$ payments to each declared valuation profile.  The mechanism proceeds by eliciting a valuation profile $\decls$ from each of the agents, called the \emph{declaration profile}.  It then applies the allocation and payment rules to $\decls$ to obtain an allocation and payment for each agent.  Crucially, we do not assume that $\decls$ is equal to $\types$.  We will write $SW(\decls)$ for $SW(\alg(\decls),\types)$ when the allocation rule and type profile are clear from context.

We will be concerned with two different payment rules, {\em first price}
and {\em critical price}. 
In a first price mechanism, an agent is charged their
declared bid $\decli(S)$ for any allocated set $S$. 
For notational convenience, we
let $\calm_1(\alg)$ denote the mechanism using allocation
rule $\alg$ and the first price payment rule. In the critical
price payment rule, an agent is charged his critical value 
$\criti(S,\declsmi)$ for any allocated set $S$.
We will let $\calm_2(\alg)$ denote the mechanism 
using allocation
rule $\alg$ and the critical price payment rule.


\subsection{Equilibria of One-shot Auctions}
The utility of agent $i$ in mechanism $\calm = (\alg,P)$, given declaration profile $\decls$ and type profile $\types$, is $\util_i^{\typei}(\decls) = \typei(\alg_i(\decls)) - P_i(\decls)$.  We will often omit the dependence on $\typei$ when it is clear from context, and write simply $\utili(\decls)$.
We say that declaration $\decli$ weakly dominates
$\decl'$ if for
all $\declsmi$, $\utili(\decli,\declsmi) \geq \utili(\decli',\declsmi)$, and that there exist at least one $\declsmi$ for which the inequality is strict.

We consider a Bayesian setting in which the true types of the agents are not fixed, but are rather drawn from a known probability distribution $\tdists$ over the set of valuation profiles. We first assume that $\tdists = \tdist_1 \times \dotsc \times \tdist_n$ is the product of independent distributions, where $\tdisti(\typei)$ is the probability that agent $i$ has type $\typei$. (Later we will also consider correlated distributions over type profiles.)
We write $SW_{opt}(\tdists)$ for $\E_{\types \sim \tdists}[SW_{opt}(\types)]$.  


A \emph{bidding strategy} for agent $i$ is a function $\bidi$ that maps a type $\typei$ to a distribution over declarations for agent $i$.
We think of $\bidi(\typei)$ as the (randomized) bidding strategy employed by agent $i$ given that his true type is $\typei$.  We will abuse notation slightly and also write $\bidi(\typei)$ for the random variable representing a declaration chosen from the corresponding distribution. We write $\bids(\types) = \bidi[1](\typei[1]) \times \dotsc \times \bidi[n](\typei[n])$ for the (distribution over) declaration profiles resulting from applying the bid functions in $\bids$ to type profile $\types$.  The strategy profile $\bids$ forms a \emph{(mixed) Bayesian Nash Equilibrium} (BNE) if, for every $i \in [n]$ and every $\typei$ in the support of $\tdisti$, agent $i$ maximizes his expected utility by making a declaration drawn from distribution $\bidi(\typei)$.  That is, for each agent $i$, each possible type $\typei$, and every distribution $\sdisti'$ over declarations,
\[
   \E_{\typesmi \sim \tdistsmi}[\utili(\bids(\types))]
     \geq
   \E_{\typesmi \sim \tdistsmi, \decli \sim \sdisti'}[\utili(\decli, \bidsmi(\typesmi))].
\]


For a mechanism $\calm = (\alg,P)$, we will write 
$SW_{\calm}(\tdists,\bids)$ to mean $\E_{\types \sim \tdists}\left[\sum_i \typei(\algi(\bids(\types)))\right]$, the expected social welfare given type distribution $\tdists$ and strategy profile $\bids$.


The \emph{(mixed) Bayesian price of anarchy} (BPoA) of mechanism $\calm$ is defined as 
\[
 \sup_{\tdists,\bids}\frac{SW_{opt}(\tdists)}{SW_{\calm}(\tdists,\bids)}
\]
where the supremum is over all type distributions $\tdists$ and mixed BNE $\bids$ for $\tdists$.  In other words, the BPoA of $\calm$ is the worst-case ratio between the expected welfare at BNE and the expected optimal welfare.


We can further extend the definition of BNE to  
allow a correlated distribution over type profiles.
The definition for correlated BNE and correlated Bayesian price of anarchy is then the same as the above definitions, where we would no longer assume that $\tdist$
is a product of independent distributions.

Returning to the case in which $\tdist$ is a product distribution, a number of special cases deserve mention.  When all type distributions are point masses (i.e., each agent's type is determined), a BNE is referred to as a \emph{(mixed) Nash Equilibrium} (NE).  The Price of Anarchy (PoA) of a mechanism $\calm$ is defined analogously to the BPoA, but with respect to fixed type profiles and mixed Nash equilibria.  It follows that the BPoA is always at least the PoA for a given mechanism.  A BNE (or NE) is called \emph{pure} if its constituent bidding strategies are deterministic.  In general a pure Nash equilibrium may not exist for a given mechanism and type profile; see Appendix \ref{sec:pureNE}.

One can generalize mixed NE by relaxing the assumption that the declaration distributions are 
independent.  That is, one might allow $\bids(\types)$ to be an arbitrary distribution over declarations, rather than a product distribution. 
%
%
A distribution $\sdists$ over declaration profiles is a 
coarse correlated equilibrium (CCE) for type profile $\types$ if, for all $i$ and all declaration distributions $\sdisti'$,

\begin{equation}
\label{eq:ccne}
\E_{\decls \sim \sdists}[\utili(\decls)] \geq \E_{\decls \sim (\sdisti',\sdistsmi)}[\utili(\decls)].
\end{equation}

Note that when the agent declaration distributions are independent, 
CCE is equivalent to mixed NE.
We define the analogous price of anarchy concepts; it follows
that the pure price of anarchy is at most the mixed price of anarchy 
which in turn is 
at most the coarse correlated price of anarchy. 

\subsection{Greedy Allocation Rules}
\label{sec:greedy}
We describe a special type of allocation rule, which we will refer to as a \emph{greedy allocation rule}.  These are motivated by the priority framework in 
Borodin, Nielsen and Rackoff \cite{BNR-02} and the monotone greedy algorithms of Mu'alem and Nisan \cite{MN-08}, extended to be adaptive as in 
Borodin and Lucier
\cite{BL-10b}. 
We begin with some definitions.  A \emph{partial allocation profile} for agents $N \subseteq [n]$ is an allocation profile $\allocs$ with $\alloci = \emptyset$ for all $i \not\in N$.  A partial allocation profile is \emph{feasible} if there is some feasible allocation profile that extends it.  Given a partial allocation profile $\allocs$ for subset $N$, some $i \not\in N$, and allocation $y \subseteq M$, we say $y$ is a \emph{feasible allocation for $i$ given $N$} if the partial allocation remains feasible when we set $\alloci = y$.

A \emph{priority function} is a function $\rank : [n] \times 2^M \times \reals \to \reals$.  We think of $r(i,S,v)$ as the priority of allocating $S \subseteq M$ to player $i$ when $\vali(S) = v$. We say that $\rank$ is monotone if it is 
non-decreasing in $v$ and monotone non-increasing in $S$ with respect to set inclusion. 

\begin{figure}
\begin{center}
	\fbox{
		\begin{minipage}{0.95\linewidth}

\textbf{Priority Algorithm}
\medskip
\hrule

\medskip

\textbf{Input:} Declaration profile $\decls = \decl_1, \dotsc, \decl_n$.

%
\begin{tabbing}
1. \= Fix a monotone priority function $r$.  Let $N = \emptyset$.\\
2. \> Repeat until $N = [n]$:\\
3. \> \quad \= Choose $i \notin N$ and feasible allocation $S \subseteq M$ for $i$ given $N$ that maximizes $r(i,S,\decli(S))$\\
4. \> \> Set $x_i = S$; add player $i$ to $N$\\
5. \> return $\allocs = (x_1, \dotsc, x_n)$
\end{tabbing}
		\end{minipage}
	}

			\caption[Priority algorithm framework]{The framework for a non-adaptive priority algorithm.}
			\label{fig.priority}
	\end{center}
\end{figure}

We consider two types of greedy allocation algorithms.  A \emph{non-adaptive greedy allocation algorithm} $\alg$ is an allocation algorithm as defined in Figure \ref{fig.priority}. We say that $\alg$ is monotone when the priority 
function $\rank$ is monotone.  
%
%
%
%
%
We assume that ties in step 3 are broken in an arbitrary but fixed manner 
(i.e. we assume that the priority function is a 1-1 function inducing 
a total ordering).  

A non-adaptive algorithm fixes a single priority function that is used throughout its execution.  By constrast, an \emph{adaptive greedy allocation algorithm} can change its priority function on each iteration, depending on the partial allocation formed on the previous iterations.  
\comment{
Note that our definition of greedy allocation rules explicitly allows only a single allocation to each agent.  This is in contrast to a very different type of ``greedy-like'' allocation rule, in which one iterates over the objects and the allocation to each agent is built up incrementally (eg. for submodular combinatorial auctions  \cite{LLN-01}).  Such incremental allocation rules are not covered by our results; we leave open the BNE analysis of their implementations.
}




\subsection{Applications}
\label{appendix:applications}

We now describe some applications of greedy algorithms for particular combinatorial allocation problems. 

\paragraph{Combinatorial Auctions}
The general combinatorial auction problem is defined by the feasibility constraint that no two allocations can intersect.  
Lehmann et al \cite{LOS-99} show that the (non-adaptive) greedy allocation rule with $r(i,S,v) = \frac{v}{\sqrt{|S|}}$ achieves a $\sqrt{2m}$ approximation ratio for CAs.


\paragraph{Cardinality-restricted Combinatorial Auctions}
In the special case that players' desires are restricted to sets of size at most $k$, the non-adaptive greedy algorithm  with $r(i,S,v) = v$ is $k$-approximate assuming single-minded agents.  This translates to a $(k+1)$ approximate algorithm for general (i.e. multi-minded) agents. 
\comment{
If $k \geq 2$ then this is a blocking allocation problem, and the first-price mechanism has a pure equilibrium and obtains a $(k+1)$ approximation at any pure equilibrium by Theorem \ref{thm.pureNE}.
}

\paragraph{Multiple-Demand Unsplittable Flow Problem}
In the unsplittable flow problem (UFP), we are given an undirected graph with edge capacities.  The objects are the edges, and each valuation function is such that agent $i$ has some value $v(s,t)$ for being given a path from $s$ to $t$.  Each agent additionally specifies a fractional demand $d_i \in [0,1]$ corresponding to a desired amount of flow to send along the given path.  An allocation is feasible if the total allocated flow along each edge is no more than its capacity.  Let $B$ be the minimum edge capacity.  A primal-dual algorithm, which is an adaptive greedy allocation rule, obtains an $O(m^{1/(B-1)})$ approximation for any $B > 1$ \cite{BKV-05}.

\paragraph{Convex Bundle Auctions}
In a convex bundle auction, $M$ is the plane $\reals^2$, and allocations must be non-intersecting compact convex sets.  We suppose that agents declare valuation functions by making bids for such sets.  Given such a collection of bids, the aspect-ratio, $R$, is defined to be the maximum diameter of a set divided by the minimum width of a set.  A non-adaptive greedy allocation rule using a geometrically-motivated priority function yields an $O(R^{4/3})$ approximation \cite{BB-04}.  Alternative greedy algorithms yield better approximation ratios for special cases, such as rectangles.

\comment{
Since this auction is a blocking allocation problem, this mechanism has pure Nash equilibria, and its price of anarchy in pure strategies is also $O(R^{4/3})$.
}

\paragraph{Max-profit Unit Job Scheduling}
In this problem, each bidder has a job of unit time to schedule on one of multiple machines.  A bidder has various windows of time of the form (release time, deadline, machine) in which his job could be scheduled, with a potentially different profit resulting from each window.  The profits and windows are private information to each bidder.  The goal of the mechanism is to schedule the jobs to maximize the total profit.  The greedy algorithm that orders bids by value obtains a $3$-approximation, and is symmetric with respect to agents and objects.  

Unlike the previous examples, for the case of single-minded bidders, 
there is an 
optimal dynamic programming algorithm that runs in time $O(n^7)$ \cite{BAP-04}. 
Since this algorithm solves the problem optimally, it is incentive compatible.  In this case, the resulting price of anarchy for the greedy algorithm is appealing primarily due to its linear runtime and simple allocation rule.

\section{Strong Loser-Independence}
\label{sec:strong-loser-indep}

Chekuri and Gamzu \cite{CG-09} introduced a property known as loser-independence for combinatorial allocation algorithms in single-parameter domains. They define an algorithm for a combinatorial allocation problem to be \emph{loser-independent} if, whenever $\alg_i(\decli, \declsmi) = \alg_i(\decli',\declsmi) = \emptyset$ for some $i$, $\declsmi$, $\decli$, and $\decli'$, then it must be that $\alg(\decli, \declsmi) = \alg(\decli',\declsmi)$.  That is, if a ``losing'' agent (i.e.\ an agent who is allocated no items) modifies his declaration in such a way that he still receives no items, this cannot affect the outcome of algorithm $\alg$.
Note that loser-independence is a condition on declaration profiles, rather than on bidding functions, since the loser-independence notion is purely algorithmic and is not a condition on equilibria.
In our results we will make use of a stronger property of greedy algorithms, which we call strong loser-independence.

\begin{defn}
An allocation rule $\alg$ is \emph{strongly loser-independent} if, whenever $\decls$ and $\decls'$ satisfy $\alg(\decls) \neq \alg(\decls')$, there exists an agent $i$ and set $S \neq \emptyset$ such that $\decli(S) \neq \decli'(S)$ and either $\alg_i(\decls) = S$ or $\alg_i(\decli',\declsmi) = S$. 
\end{defn}

Roughly speaking, if $\alg$ is a strongly loser-independent algorithm, then whenever a valuation profile changes from $\decls$ to $\decls'$ via modifications to ``losing bids'' (i.e.\ an agent $i$'s declared value for sets that are not allocated to him, when others bid according to $\declsmi$), algorithm $\alg$ will return the same outcome on inputs $\decls$ and $\decls'$.  We note that our definition requires that either $\alg_i(\decls) = S$ or $\alg_i(\decli',\declsmi) = S$, rather than $\alg_i(\decls') = S$.  The intuition is that we think of ``losing bids'' as being losers with respect to the original declaration profile $\decls$.

The property of strong loser-independence strengthens the definition of loser-independence due to Chekuri and Gamzu in two ways.  First, we extend from single-parameter settings to multiple-parameter settings by considering losing bids rather than losing agents.  Second, we require that the algorithm outcome be unaffected if multiple agents simultaneously modify losing bids.

It is clear from the definitions that all strongly loser-independent algorithms are loser-independent (i.e.\ by considering the case when $\decls$ and $\decls'$ differ only on the declaration of a single agent).  However, not all loser-independent algorithms are strongly loser-independent, even in single-minded domains.  For example, consider the combinatorial auction problem and suppose that $\alg$ is an algorithm that optimizes social welfare exactly and breaks ties consistently.  Then $\alg$ is loser-independent, since a losing agent's bid does not affect the optimal allocation.  However, $\alg$ is not strongly loser-independent, as the following instance shows.  Consider an auction of two items $\{a,b\}$ to three bidders.  If the (single-minded) bidder declarations are $\decli[1](\{a,b\}) = 10$, $\decli[2](\{a\}) = 3$, and $\decli[3](\{b\}) = 3$, then the outcome is that agent $1$ wins his desired set.  On the other hand, if the bidder declarations are given by $\decls'$ where $\decli[1]'(\{a,b\}) = 10$, $\decli[2]'(\{a\}) = 6$, and $\decli[3]'(\{b\}) = 6$, then the outcome changes: agents $2$ and $3$ win their desired sets.  However, $\decls$ and $\decls'$ do not differ in their declarations for any sets allocated by $\alg(\decls)$, $\alg(\decli[1]',\declsmi)$, $\alg(\decli[2]',\declsmi)$, or $\alg(\decli[3]',\declsmi)$, as agent $1$ wins his desired set in each of these four cases.  This contradicts the definition of strongly loser-independence. 


As we now show, all greedy algorithms satisfy the strong loser-independence property. 

\begin{lem}
\label{lem.greedy1}
Every (monotone) adaptive greedy algorithm is ($\mon$) and strongly loser-independent.
\end{lem}
\begin{proof}
The monotonicity property follows immediately when  
the priority function in the greedy algorithm is a monotone
function. 
 
Let $\alg$ be an adaptive greedy allocation rule, and choose any $\decls$ and $\decls'$ such that $\alg(\decls) \neq \alg(\decls')$.  We will show that there exists some $i$ and $S$ such that $\decli(S) \neq \decli'(S)$ and either $\alg_i(\decls) = S$ or $\alg_i(\decli',\declsmi) = S$.

Recall the definition of an adaptive greedy algorithm, and consider the iterations of $\alg$ on inputs $\decls$ and $\decls'$.  Let $k$ be the first iteration in which the allocation of $\alg$ differs on these two inputs.  Suppose that $\alg$ allocates set $U$ to agent $\ell$ on iteration $k$ when the input is $\decls$, and allocates $T$ to agent $j$ on iteration $k$ when the input is $\decls'$.  

For each iteration $q < k$, write $i_q$ for the agent allocated to by $\alg$ (on either input profile) and $S_q$ for the set allocated to $i_q$.  Note that if $\decli[i_q](S_q) \neq \decli[i_q]'(S_q)$ for any $q < k$ then we have the desired result with $i = i_q$ and $S = S_q$.  We can therefore assume that $\decli[i_q](S_q) = \decli[i_q]'(S_q)$ for all $q < k$.  This implies that the bids resolved by $\alg$ are identical on all iterations preceeding $k$ on inputs $\decls$ and $\decls'$, and therefore the values of ranking functions used in each iteration up to $k$ must be identical for inputs $\decls$ and $\decls'$. 
Write $r_q$ for the ranking function used in iteration $q$ for each $q \leq k$.  Thus, since the allocation on iteration $k$ changed from choosing set $U$ for agent $\ell$ to choosing set $T$ for agent $j$, it must be that either $r_k(\ell,U,\decli[\ell](U)) \neq r_k(\ell,U,\decli[\ell]'(U))$ or $r_k(j,T,\decli[j](T)) \neq r_k(j,T,\decli[j]'(T))$.  This implies that either $\decli[\ell](U) \neq \decli[\ell]'(U)$ or $\decli[j](T) \neq \decli[j]'(T)$.

If $\decli[\ell](U) \neq \decli[\ell]'(U)$ then we have the desired result with $i = \ell$ and $S = U$, since $\alg_\ell(\decls) = U$.  We can therefore assume that $\decli[\ell](U) = \decli[\ell]'(U)$ and $\decli[j](T) \neq \decli[j]'(T)$.  Consider now the behaviour of algorithm $\alg$ on input $(\decli[j]',\declsmi)$.  We claim that $\alg_j(\decli[j]',\declsmi) = T$.  Note that this implies the desired result with $i = j$ and $S = T$.  To prove the claim, recall that $\decli[i_q](S_q) = \decli[i_q]'(S_q)$ for all $q < k$.  Thus, for each $q < k$ and each feasible set $S$ that could be allocated to agent $j$ on iteration $q$,
\[ r_q(i_q,S_q,\decli[i_q](S_q)) =  r_q(i_q,S_q,\decli[i_q]'(S_q)) > r_q(j,S,\decli[j]'(S)) \] 
since $\alg$ allocates $S_q$ to $i_q$ on input $\decls'$.  We conclude that, on input $(\decli[j]',\declsmi)$, $\alg$ allocates $S_q$ to agent $i_q$ on each iteration $q < k$ .  On iteration $k$, we have $r_k(j,T,\decli[j]'(T)) > r_k(j,T',\decli[j]'(T'))$ for any feasible $T' \neq T$ (since $\alg$ allocates $T$ to $j$ on input $\decls'$) and 
\[ r_k(j,T,\decli[j]'(T)) > r_k(\ell,U,\decli[\ell]'(U)) = r_k(\ell,U,\decli[\ell](U)) \geq r_k(i,S,\decli(S)) \]
for any feasible $i \neq j$, due to our assumption that $\decli[\ell]'(U) = \decli[\ell](U)$ and the fact that $\alg$ allocates $U$ to $\ell$ on iteration $k$ for input $\decls$.  We therefore conclude that 
$\alg_j(\decli[j]',\declsmi) = T$ as required. 
\end{proof}

We next explore an implication of a strongly loser-independent algorithm $\alg$ being a worst-case $c$-approximation.  If $\alg$ is a $c$-approximate algorithm, then (on any input) the sum of the declared values for its output profile approximates the sum of the declared values for the optimal allocation.  We now show that it also approximates the sum of the critical values of the optimal allocation profile.

\begin{lem}
\label{lem.greedy}
If $\alg$ is a $c$-approximate strongly loser-independent algorithm, then for any type profile $\vals$ and allocation profile $\mathbf{y}$, $\sum_{i \in [n]}\vali(\alg_i(\vals)) \geq \frac{1}{c} \sum_{i \in [n]} \criti(y_i, \valsmi)$.
\end{lem}
\begin{proof} 
Choose any $\eps > 0$.  For all $i$, let $\vali'$ be the single-minded declaration for set $y_i$ at value $\criti(y_i, \valsmi) - \eps$.  Let $\vali^*$ be the pointwise maximum of $\vali'$ and $\vali$.  That is, for all $S \subseteq M$, $\vali^*(S) = \max\{\vali(S), \vali'(S)\}$.  By definition of critical prices, we have that $\alg_i(\vali^*,\valsmi) = \alg_i(\vals)$ for all $i$, and furthermore $\vali^*(\alg_i(\vals)) = \vali(\alg_i(\vals))$.  Since $\alg$ is strongly loser-independent, we must therefore have $\alg(\vals) = \alg(\vals^*)$.
Since $\alg$ is a $c$-approximation, we conclude that $SW(\allocs(\vals), \vals) = SW(\allocs(\vals^*), \vals^*) \geq \frac{1}{c} SW(\mathbf{y},\vals^*) \geq \frac{1}{c} \sum_{i \in [n]} \criti(y_i, \valsmi) - n\eps$.
The result follows by taking the limit as $\eps \to 0$.
\end{proof}

\comment{
\allan{Now the following paragraph doesn''t make sense. Maybe this
all should go as an open problem in the conclusions unless we can come up with a simple example of a ($\mon$) $\sli$  algorithm that is not (monotone) greedy.}
Since this strongly loser independent property plays such a predominant
role in our development, it is natural to ask if the converse of 
Lemma \ref{lem.greedy1} holds; that is, if 
every strongly loser
independent  rule is equivalent to a greedy allocation rule. However, 
it is not hard to see that strongly loser-independent rules
need not even be $\mon$ and hence cannot be a greedy allocation rule 
within the
formulation given in section \ref{sec:greedy}. 
However, this leaves open the possibility that every monotone strongly
loser independent rule is a greedy rule.  

\begin{example}
Even in the case of one single-minded agent, a strongly loser-independent 
rule might not be monotone. Consider two valuations $v^1(a) = \epsilon$ and 
$v^2(a) = 1$ and let $\alg(v^1) = \{a\}$ and $\alg(v^2) = \emptyset$. Then 
$\alg$ is strongly loser-independent but clearly not monotone.
\end{example}

\allan{Adding a footnote. If we think this is all too confusing, then we could go back and redo everything to be monotone. But I prefer isolating out where monotinicity is not needed.}
}
For brevity, for the remainder of this paper we will say ``monotone strongly loser-independent'' to mean both strongly loser-independent and component-wise monotone\footnote{For pure Nash equlibirium price of anarchy results, if we assume no over-bidding, we do not need monotonicity but it is necessary for all of our other results.}.

\subsection{Applying Strong Loser Independence}
\label{sec:applying-loser-indep}

Strong loser-independence is a strictly algorithmic concept,
devoid of game theoretic considerations. Our general approach 
will be to derive price of anarchy results for any mechanism that 
uses a strongly loser-independent $c$-approximation $\alg$ as its allocation algorithm. 
To do so, we will be using Lemma~\ref{lem.greedy} in conjuction with the
assumption that a given bid profile is an equilibrium. 

At a high level, our argument will be as follows.
For each pricing rule and equilibrium concept, equilibrium will imply an inequality of the form 
\[ v_i(y_i) \leq \lambda \cdot \criti(y_i, \declsmi) + \mu \cdot v_i(\alloci(\decls))\] 
where $\mathbf{y}$ is an optimal allocation. 
(For Bayesian equilibrium, these terms will be expectations.)  This allows us to charge the optimal gain for each 
agent to its critical value and its welfare from the algorithm. We then exploit
Lemma \ref{lem.greedy} to 
convert this bound into a relationship between the optimal welfare and the welfare at equilibrium.
To make this
more specific, in our pure Nash equilibrium result for a first price
mechanism (Theorem~\ref{thm.pureNE}), we show the following 
(somewhat stronger) inequality:
$$\typei(y_i) \leq \criti(y_i,\declsmi) + \typei(\alloci(\decls)) - \decli(\alloci(\decls))$$ 
It will then follow   
(ignoring an  $n \epsilon$ term which disappears as $\epsilon \rightarrow 0$) that:
\begin{align*}
  \sum_i \typei(y_i)
  & \leq \sum_i \criti(y_i,\declsmi) + \sum_i \typei(\alloci(\decls)) - \sum_i \decli(\alloci(\decls))\\
  & \leq (c-1)\sum_i \decli(\alloci(\decls)) + \sum_i \typei(\alloci(\decls))  
\end{align*}


In other words, the high-level approach is to charge an agent's welfare in the optimal outcome against his welfare at equilibrium plus the welfare of other ``price-setting'' agents.  
This approach is similar 
to the  smoothness argument
as formulated by Syrgkanis and Tardos \cite{SyrgkanisT13}.
However, there is a difference in our approaches. 
The smoothness condition in \cite{SyrgkanisT13} is tailored to allocation mechanisms and asserts the existence of some $\decl_i$ (for each player $i$) 
satisfying such
an inequality whereas we are assuming that $\decls$ is an equilibrium. 
The benefit of their immediate reduction to smoothness is that their 
price of anarchy results for pure equilibria carry over 
immediately to Bayesian price of anarchy. However, this prohibits 
establishing certain tight bounds; for example, in the first price mechanism 
we show that the pure price of anarchy is $c$, which cannot be achieved
via smoothness since this bound does not hold for the Bayesian price of anarchy.

\newcommand{\udecli}{\overline{\decli}}
\newcommand{\ubidi}{\overline{\bidi}}


\section{First-Price Mechanisms}
\label{sec:mixedBNE}

In this section we analyze greedy algorithms paired with a first price payment scheme. More precisely (with the exception of results relating 
to correlated Bayesian equilibirium where we will consider more 
specific greedy allocations), given a 
strongly loser-independent algorithm $\alg$,
%
we will be studying the performance of the first-price mechanism $\calm_1(\alg)$ at equilibrium.

Our first step will be to show that a utility-maximizing declaration of an agent never involves overbidding on a set that he may possibly be allocated.  This will imply that agents do not employ overbidding strategies at equilibrium.
It may appear at first glance that any strategy that recommends overbidding on sets is obviously dominated for any allocation algorithm, since winning any bid larger than one's true value leads to negative utility.  However, we must also show that an agent cannot find it advantageous to overbid on some set $S$ in order to affect his probability of winning some other set $T$.  We will demonstrate that such situations cannot occur when allocations are chosen by a 
strongly loser independent algorithm.


For a type $\typei$ and a declaration $\decli$, we will write $\udecli$ for the declaration defined as $\udecli(S) = \min\{\typei(S), \decli(S)\}$.  That is, $\udecli$ agrees with $\decli$, except that the declared value of each set can be at most the true value for that set.  Note that $\udecli = \decli$ precisely if $\decli$ does not overbid on any set.

We now show that any declaration $\decli$ that overbids on a set 
that could potentially be won is  
weakly dominated by strategy $\udecli$.  


\begin{lem}
\label{lem.firstprice.overbid}
For any $\msli$ $\alg$, valuation $\vali$, and declaration profile $\decls$, we have
$\utili(\decls) \leq \utili(\udecli,\declsmi).$
Moreover, the inequality is strict when $\decli(\alg(\decls)) > \typei(\alg(\decls))$.
\end{lem}

\begin{proof}
Let $S = \alg_i(\decls)$.  Suppose first that $\decli(S) > \typei(S)$. Then $\utili(\decls) = \typei(S) - \decli(S) < 0$.  Since $\typei(T) - \udecli(T) \geq 0$ for every set $T$, this implies that $\utili(\udecli,\declsmi) > \utili(\decls)$, as required.

Next suppose that $\decli(S) \leq \typei(S)$, so that $\udecli(S) = \decli(S)$.  We claim that $\algi(\udecli,\declsmi) = S$.  Suppose not, for contradiction. Then we can construct a sequence of declarations $(\decl^1, \decl^2, \dotsc, \decl^k)$, with $\decl^1 = \decli$ and $\decl^k = \udecli$, such that adjacent declarations differ only on a single set and declared values only decrease.  Suppose $j$ is minimal such that $\algi(\decl^j, \declsmi) \neq S$; such a $j > 1$ must exist since, by assumption, $\algi(\udecli,\declsmi) \neq S$.  Then (a) $\decl^{j-1}$ and $\decl^j$ differ only on the value assigned to some set $T$, (b) $\decl^{j-1}(T) > \decli^{j}(T)$, (c) $\algi(\decl^{j-1},\declsmi) = S$, and (d) $\algi(\decl^j,\declsmi) \neq S$.  Strong loser-independence then implies that $\algi(\decl^{j},\declsmi) = T$.  However, the fact that $\decl^{j-1}(T) > \decl^j(T)$ then contradicts the component-wise monotonicity of $\alg$.

We conclude by contradiction that $\algi(\udecli,\declsmi) = S$.  Since $S$ is also $\alg(\decls)$, we have $\utili(\decls) = \vali(S) - \decli(S) = \utili(\udecli,\declsmi)$ as required.  

\end{proof}


An immediate corollary is that if $\bidi$ is a bidding strategy, and there exists a type $\typei$
and set $S$ such that $(\bidi(\typei))(S) > \vali(S)$, then $\bidi$ is weakly dominated by the
strategy $\ubidi$.  Moreover, $\ubidi$ is strictly better, in terms of utility,
under any distribution of declarations in which agent $i$ wins set $S$ with positive probability.
We conclude that at any BNE of mechanism $\mech_1(\alg)$, no player will overbid on a set
that he wins with positive probability.

\begin{cor}
\label{cor.firstprice.overbid}
For any $\msli$ $\alg$, BNE $\bids$, type $\typei$, and set $S$, if
$\Pr_{\valsmi \sim \distsmi}[\alg_i(\bids(\vals)) = S] > 0$ then $(\bidi(\typei))(S) \leq \typei(S)$.
\end{cor}

\subsection{Pure Nash Equilibria}

We are now ready to bound the price of anarchy of $\calm_1(\alg)$.  
We begin with a result for pure Nash equilibria, rather than the fully general BNE case.  

\begin{thm}
\label{thm.pureNE}
Suppose $\alg$ is a  
$c$-approximate 
monotone strongly loser independent allocation rule
for a combinatorial 
allocation problem.  Then the price of anarchy of $\mech_1(\alg)$ is 
at most $c$.
\end{thm}

\begin{proof}
Fix type profile $\types$ and suppose that $\bids$ forms a pure Nash equilibrium.  Since the Nash equilibrium is pure, we will write $\decls = \bids(\types)$ for notational convenience.  Let $\mathbf{y}$ be an optimal allocation for $\types$, and let $\allocs(\cdot)$ denote the allocation rule for $\alg$.  Lemma \ref{lem.greedy} implies
\begin{equation}
\label{eq.purene.1}
\sum_i \decli(\alloci(\decls)) 
\geq \frac{1}{c}\sum_i \criti(y_i,\declsmi).
\end{equation}
%
Choose arbitrarily small $\eps > 0$ and let $\decli'$ be the single-minded declaration for set $y_i$ at value $\criti(y_i,\declsmi)+\epsilon$.  Then $\alloci(\decli',\declsmi) = y_i$ (from the definition of critical values) and hence $\utili(\decli',\declsmi) = \typei(y_i) - \criti(y_i,\declsmi)-\eps$.  Since $\decls$ is a Nash equilibrium, it must be that 
\begin{align*}
   \typei(y_i) - \criti(y_i,\declsmi) - \eps 
   & = \utili(\decli',\declsmi) \\
   & \leq \utili(\decli,\declsmi) \\
   & = \typei(\alloci(\decls)) - \decli(\alloci(\decls)).
\end{align*}
Summing over all $i$ and applying \eqref{eq.purene.1} and Corollary \ref{cor.firstprice.overbid}
we have 
\begin{align*}
  \sum_i \typei(y_i) 
  & \leq \sum_i \criti(y_i,\declsmi) - \sum_i \decli(\alloci(\decls)) + \sum_i \typei(\alloci(\decls)) + n\eps \\
  & \leq (c-1)\sum_i \decli(\alloci(\decls)) + \sum_i \typei(\alloci(\decls)) + n\eps \\
  & \leq c\sum_i \typei(\alloci(\decls)) + n\eps
\end{align*}
which, taking $\eps \to 0$, implies
\begin{align*}
    SW(\allocs(\decls),\types)
    & = \sum_i \typei(\alloci(\decls)) \\
    & \geq \frac{1}{c}\sum_i \typei(y_i) \\ 
    & = \frac{1}{c}SW_{OPT}(\types)
\end{align*}
as required.
\end{proof}

The power of Theorem \ref{thm.pureNE} is marred by the fact that,
for some problem instances, the mechanism $\mech_1(\alg)$ is not
guaranteed to have a pure Nash equilibrium.
An example is given in Appendix \ref{sec:pureNE}.


\subsection{Bayes-Nash Equilibria}

\newcommand{\barutil}{\overline{\utili}}
\newcommand{\barcrit}{\overline{\criti}}

We are now ready to bound the mixed Bayesian price of anarchy for mechanism $\calm_1(\alg)$. 

\begin{thm}
\label{thm.mixedNE}
Suppose $\alg$ is a $\msli$ 
for a combinatorial allocation problem.  Then the Bayesian price of anarchy of $\mech_1(\alg)$ is at most\footnote{
In the initial conference version of this work, we presented a bound of $c+O(\log c)$ on the BPOA. Subsequently, this bound was independently improved by Lucier \cite{L-thesis11} to $c + O(c^2/e^c)$ and by  
Syrgkanis and Tardos \cite{SyrgkanisT13} to $c + O(c/e^c)$.  We present here a slightly modifiied version of the argument from Lucier \cite{L-thesis11}, which yields the improved Syrgkanis and Tardos \cite{SyrgkanisT13} bound of $c + O(c/e^c)$.} 
$\frac{c}{1-e^{-c}}$ for every independent type distribution $\dists$.
\end{thm}

We note that  $\frac{c}{1-e^{-c}} \leq c\left(1+\frac{2}{e^c}\right) = c + O(c/e^c)$.  The remainder of this subsection is dedicated to the proof of Theorem \ref{thm.mixedNE}.

Fix a product distribution $\dists$ over type profiles and let $\bids(\cdot)$ be a (possibly mixed) Bayes-Nash equilibrium with respect to $\dists$.  
Choose some type declaration $\types$ and let $\mathbf{y}^\types$ denote an optimal allocation for $\types$.
Following the proof of Theorem \ref{thm.pureNE}, we would like to bound the expected value of $\criti(y_i^\types,\declsmi)$ with respect to $\typei(y_i^{\types})$ and $\utili(\bids(\types))$ for each $i$.  We encapsulate this bound in Lemma \ref{lem.bayes.k} and Corollary \ref{cor.bayes.k}, below.  This will allow us to use Lemma \ref{lem.greedy} to obtain a relation between the expected welfare of $\alg$ and the expected optimal welfare; this relationship is given in Lemma \ref{lem.bayes.bound}.

\begin{lem}
\label{lem.bayes.bound}
Suppose that $\alg$ is a $c$-approximate 
monotone strongly loser independent allocation rule
and that there exist constants $\gamma \geq 0$ and $\sigma_i \in [0,c]$ for $i \in [n]$ such that, whenever $\bids$ is a Bayes-Nash equilibrium for $\mech_1(\alg)$, it is the case that for all $i$, all $\typei$, and all $S \subseteq M$,
\begin{align*}
   \E_{\typesmi}[\criti(S,\bidsmi(\typesmi))] & \geq \gamma\typei(S) - \sigma_i \E_{\typesmi}[\utili(\bids(\types))].
\end{align*}
Then $\E_{\types}[SW(\alg(\bids(\types)),\types)] \geq \frac{\gamma}{c}\E_{\types}[SW_{OPT}(\types)]$.
\end{lem}

\begin{lem}
\label{lem.bayes.k}
Suppose that $\bids$ is a Bayes-Nash equilibrium for mechanism $\mech_1(\alg)$ and distribution $\dists$.  Then for all $i$, all $\typei$, and all $S \subseteq M$,
\begin{align*}
   \E_{\typesmi}[\criti(S,\bidsmi(\typesmi))] \geq 
       \typei(S) - \left(1+\ln \frac{\typei(S)}{\E_{\typesmi}[\utili(\bids(\types))]}\right)\E_{\typesmi}[\utili(\bids(\types))].
\end{align*}
\end{lem}

Before proving Lemmas \ref{lem.bayes.bound} and \ref{lem.bayes.k}, let us show how they imply Theorem \ref{thm.mixedNE}.  We first note the following simple corollary of Lemma \ref{lem.bayes.k}.

\begin{cor}
\label{cor.bayes.k}
Suppose that $\bids$ is a Bayes-Nash equilibrium for mechanism $\mech_1(\alg)$ and distribution $\dists$.  Then for all $i$, all $\typei$, and all $S \subseteq M$,
\begin{align*}
   \E_{\typesmi}[\criti(S,\bidsmi(\typesmi))] \geq (1 - e^{-c}) \cdot \typei(S) - c \cdot \E_{\typesmi}[\utili(\bids(\types))].
\end{align*}
\end{cor}
\begin{proof}
Fix agent $i$.  By Lemma \ref{lem.bayes.k}, we know
\begin{align}
\label{eq.cor.bayes}
   \E_{\typesmi}[\criti(S,\bidsmi(\typesmi))] \geq 
       \typei(S) - \left(1+\ln \frac{\typei(S)}{\E_{\typesmi}[\utili(\bids(\types))]}\right)\E_{\typesmi}[\utili(\bids(\types))].
\end{align}
Note that if $\left(1+\ln \frac{\typei(S)}{\E_{\typesmi}[\utili(\bids(\types))]}\right) \leq c$ then \eqref{eq.cor.bayes} immediately implies the desired result.  We can therefore assume otherwise, and choose $\alpha > 0$ such that 
\[ \left(1+\ln \frac{\typei(S)}{\E_{\typesmi}[\utili(\bids(\types))]}\right) = c + \alpha. \] 
Rearranging, we get that $\typei(S) = e^\alpha \cdot e^{c-1} \cdot \E_{\typesmi}[\utili(\bids(\types))]$.  Applying these two equalities to \eqref{eq.cor.bayes}, we have
\begin{align*}
\E_{\typesmi}[\criti(S,\bidsmi(\typesmi))] & \geq \typei(S) - (c + \alpha) \cdot \E_{\typesmi}[\utili(\bids(\types))] \\
& = \typei(S) - \frac{\alpha}{e^\alpha}\cdot \frac{\typei(S)}{e^{c-1}} - c \cdot \E_{\typesmi}[\utili(\bids(\types))].
\end{align*}
Since $\frac{\alpha}{e^{\alpha}}$ achieves its maximum value of $1/e$ at $\alpha = 1$, we can conclude that
\[ \E_{\typesmi}[\criti(S,\bidsmi(\typesmi))] \geq \typei(S) - \frac{1}{e^c}\cdot \typei(S) - c \cdot \E_{\typesmi}[\utili(\bids(\types))] \]
as required.
\end{proof}


Theorem \ref{thm.mixedNE} follows directly from Corollary \ref{cor.bayes.k} and Lemma \ref{lem.bayes.bound}.
%
%
We next complete the proof of Theorem \ref{thm.mixedNE} by proving Lemmas \ref{lem.bayes.bound} and \ref{lem.bayes.k}.

\begin{proofof}{Lemma \ref{lem.bayes.bound}}
Fix distribution $\dists$ over type profiles and let $\bids(\cdot)$ be a (possibly mixed) Bayes-Nash equilibrium with respect to $\dists$.  
Choose some type declaration $\types$ and let $\mathbf{y}^\types$ denote an optimal allocation for $\types$.
We know that for all $i \in [n]$ and $\types$,
\begin{align*}
\E_{\typesmi'}[\criti(y_i^\types,\bidsmi(\typesmi'))] & \geq \gamma\typei(y_i^\types) - \sigma_i \E_{\typesmi'}[\utili(\bidi(\typei), \bidsmi(\typesmi')))].
\end{align*}
Note the distinction between $\typesmi'$, over which we are taking expectations, and $\typesmi$, which is the type profile fixed to define $y_i^\types$.  Now, summing over $i$ and taking expectation over all choices of $\types$, we have
\begin{equation}
\begin{split}
\label{eq.bne.main}
\E_{\types}\left[\sum_i\E_{\typesmi'}[\criti(y_i^\types,\bidsmi(\typesmi'))]\right]
\geq \gamma\E_{\types}\left[\sum_i\typei(y_i^\types)\right] 
- \E_{\types}\left[\sum_i \sigma_i \E_{\typesmi'}[\utili(\bidi(\typei),\bidsmi(\typesmi'))]\right].
\end{split}
\end{equation}
We now consider each of the three terms in \eqref{eq.bne.main}.  First, note that
\begin{equation}
\label{eq.bne.main.1}
\E_{\types}\left[\sum_i\typei(y_i^{\types})\right] = \E_{\types}[SW_{OPT}(\types)].
\end{equation}
Additionally,
\begin{equation}
\label{eq.bne.main.2}
\begin{split}
\E_{\types}\left[\sum_i \sigma_i \E_{\typesmi'}[\utili(\bidi(\typei),\bidsmi(\typesmi'))]\right]
& = \sum_i \sigma_i \E_{\types,\typesmi'}[\utili(\bidi(\typei),\bidsmi(\typesmi'))] \\ 
& = \E_{\types}\left[\sum_i \sigma_i \utili(\bids(\types))\right] \\ 
& = \E_{\types}\left[\sum_i \sigma_i \typei(\alloci(\bids(\types)))\right] - \E_{\types,\decls = \bids(\types)}\left[\sum_i \sigma_i \decli(\alloci(\bids(\types)))\right]
\end{split}
\end{equation}
where the final equality follows from the fact that our mechanism employs a first price payment scheme.  Finally,
\begin{equation}
\label{eq.bne.main.3}
\begin{split}
& \E_{\types}\left[\sum_i\E_{\typesmi'}[\criti(y_i^\types,\bidsmi(\typesmi'))]\right] \\
& \quad= \E_{\types,\types'}\left[\sum_i\criti(y_i^\types,\bidsmi(\typesmi'))\right] \ \ (\text{type independence}) \\
& \quad \leq c\E_{\types,\types',\decls' = \bids(\types')}\left[\sum_i\decli'(\alloci(\decls'))\right] \ \ (\text{Lemma \ref{lem.greedy}}) \\
& \quad = c\E_{\types,\decls = \bids(\types)}\left[\sum_i\decli(\alloci(\decls))\right] 
\end{split}
\end{equation}
Where the final equality follows from a change of variables, since $\vals$ does not appear inside the expectation on the previous line.
Substituting \eqref{eq.bne.main.1}, \eqref{eq.bne.main.2}, and \eqref{eq.bne.main.3} into \eqref{eq.bne.main}, we conclude that
\begin{align*}
c \E_{\types}\left[\sum_i\decli(\alloci(\decls))\right] \geq 
\gamma\E_{\types}[SW_{OPT}(\types)] - \E_{\types}\left[\sum_i \sigma_i \typei(\alloci(\bids(\types)))\right] 
+ \E_{\types,\decls=\bids(\types)}\left[\sum_i \sigma_i \decli(\alloci(\bids(\types)))\right]
\end{align*}
and hence
\begin{align*}
\gamma\E_{\types}[SW_{OPT}(\types)]
& \leq \E_{\types}\left[\sum_i \sigma_i \typei(\alloci(\bids(\types)))\right] + \E_{\types, \decls = \bids(\vals)}\left[\sum_i(c-\sigma_i)\decli(\alloci(\decls))\right] \\
& \leq \E_{\types}\left[\sum_i \sigma_i \typei(\alloci(\bids(\types)))\right] + \E_{\types}\left[\sum_i(c-\sigma_i)\typei(\alloci(\bids(\vals)))\right] \\
& = \E_{\types}\left[\sum_i c \typei(\alloci(\bids(\types)))\right] \\
& = c\E_{\types}[SW(\alg(\bids(\types)),\types)] 
\end{align*}
where in the second inequality we used Corollary \ref{cor.firstprice.overbid} plus the fact that $(c - \sigma_i) \geq 0$ for all $i$.  Rearranging yields
\begin{align*}
\E_{\types}[SW(\alg(\bids(\types)),\types)] \geq \frac{\gamma}{c}\E_{\types}[SW_{OPT}(\types)]
\end{align*}
as required.
\end{proofof}

\begin{proofof}{Lemma \ref{lem.bayes.k}}
Fix any $i$, $\typei$, and $S$.  
%
Since $\criti(S,\declsmi) \geq 0$ for all $\declsmi$, we have that
\begin{align*}
\E_{\typesmi}[\criti(S,\bidsmi(\typesmi))] 
& \geq \int_0^{\vali(S)} \Pr[ \criti(S,\bidsmi(\typesmi)) > z] dz \\
& = \vali(S) - \int_0^{\vali(S)} \Pr[ \criti(S,\bidsmi(\typesmi)) \leq z] dz.
\end{align*}
Recall that $\bidi(\typei)$ must maximize the expected utility of agent $i$.  Choose any $z \geq 0$ and consider the alternative strategy $\decli$ which places a single-minded bid of $z$ on set $S$.  Then since $\bidi(\typei)$ is an optimal strategy, we have that 
\begin{align*}
\E_{\typesmi}[\utili(\bids(\types))] 
& \geq \E_{\typesmi}[\utili(\decli,\bidsmi(\typesmi))] \\
& = (\vali(S) - z) \Pr[ \criti(S,\bidsmi(\typesmi)) \leq z ] 
\end{align*}
where the equality follows since any single minded bid above the critical value for $S$ insures that $S$ will be won, as a consequence of monotonicity.  We conclude that 
\[ \Pr[ \criti(S,\bidsmi(\typesmi)) \leq z ] \leq \frac{\E_{\typesmi}[\utili(\bids(\types))] }{(\vali(S) - z)}\] 
for all $0 \leq z < \vali(S)$. We also know that $\Pr[ \criti(S,\bidsmi(\typesmi)) \leq z ] \leq 1$ for all $z$.   Write $r = \vali(S) - \E_{\typesmi}[\utili(\bids(\types))]$.  We then conclude that
\begin{align*}
\E_{\typesmi}[\criti(S,\bidsmi(\typesmi))] 
& \geq \vali(S) - \int_0^{r} \frac{\E_{\typesmi}[\utili(\bids(\types))] }{(\vali(S) - z)} dz - \int_{r}^{\vali(S)} 1 dz \\
& = \vali(S) - \E_{\typesmi}[\utili(\bids(\types))] \int_{\E_{\typesmi}[\utili(\bids(\types))]}^{\vali(S)} \frac{1}{y} dy - \E_{\typesmi}[\utili(\bids(\types))] \\
& = \vali(S) - \left(1+\ln \frac{\typei(S)}{\E_{\typesmi}[\utili(\bids(\types))]}\right)\E_{\typesmi}[\utili(\bids(\types))]
\end{align*}
as required.
\end{proofof}

\begin{cor}{(of proof)}
\label{cor:coarse}
The same bound on price of anarchy applies to coarse correlated equilibrium. 
\end{cor}

\begin{proof} 
That such price of anarchy  bounds can be applied to coarse correlated 
equlibria in the full information setting  was initially observed by Roughgarden \cite{ROU-09}. Specifically,  
in the proof of Theorem~\ref{thm.mixedNE}, all occurences of 
$\E_{\types,\decls=\bids(\types)}$ can be replaced by 
$\E_{\decls \sim (\decli',\sdistsmi)}$, resulting in a bound on the coarse correlated PoA.
\end{proof} 


It may be tempting to conjecture that the (exponentially small) loss in approximation factor in Theorem \ref{thm.mixedNE} is simply an artifact of the analysis, and that the Bayesian price of anarchy of $\mech_1(\alg)$ is actually $c$.  However, we now show by way of an example that this loss is necessary;
that is, there exist instances in which the 
mixed price of anarchy (and hence the Bayesian price of anarchy) is strictly greater than $c$. 


\begin{prop}
\label{prop.mixedNE.LB}
For any $c \geq 2$, there is a combinatorial allocation problem $\calp$ and a non-adaptive greedy algorithm $\alg$ such that $\alg$ is a $c$-approximation for $\calp$, and the mixed price of anarchy for $\calm_1(\alg)$ is 
at least $c + c^2/e^{4c} = c + \Om(c^2/e^{4c})$. 
\end{prop}
\begin{proof}
We begin by describing our combinatorial allocation problem.  We choose a parameter $k > c$ that will be fixed later.  Our auction has $ck + k$ objects, which we label $a_{ij}$ for $i \in [k], j \in [c]$ and $b_i$ for $i \in [k]$.  There are $4k$ agents, labelled $A_i, B_i, C_i$, and $D_i$ for $i \in [k]$.  Our feasibility constraints are as follows.  Each agent $B_i$ or $C_i$ can receive only set $\{a_{i1}\}$ or $\emptyset$.  Each agent $D_{i}$ can receive set $\{a_{i1},a_{k1}\}$ or $\emptyset$.  Each agent $A_i$ can receive either set $\{a_{i1},a_{i2},\dotsc,a_{ic}\}$, set $\{b_i\}$, or $\emptyset$.  Under these restrictions, an allocation is feasible if each object is assigned to at most one agent.

Let $\alg$ be the non-adaptive greedy algorithm that orders bids by density: i.e.\ with priority function $r(i,S,v) = v/|S|$ when $S$ is a feasible set for agent $i$.  We claim that when $c \geq 2$, this algorithm obtains a $c$-approximation for the above combinatorial auction.  To see this, note that the (unique) set that can be allocated to any agent $B_i$, $C_i$, or $D_i$ intersects sets of size at most $c$ times larger, so if the greedy algorithm allocates to one of these agents for a value of $v$, the total value of intersecting sets in the optimal solution is at most $cv$.  On the other hand, if the greedy algorithm allocates $\{b_i\}$ to agent $A_i$, this conflicts only with the allocation of set $\{a_{i1}, \dotsc, a_{ic}\}$ to agent $A_i$, which again has value at most $c$ times greater.  Finally, suppose that the greedy algorithm allocates set $\{a_{i1}, \dotsc, a_{ic}\}$ to agent $A_i$, say with value $vc$ (i.e.\ value density $v$). This allocation can conflict only with a single allocation to an agent $B_i, C_i$, or $D_i$ plus an allocation of $\{b_i\}$ to agent $A_i$, which comprises a total of at most $3$ objects.  Since the greedy algorithm allocates by density, the total value of the conflicted bids is at most $3v$.  Since $c \geq 2$, we conclude that the allocation of $\{a_{i1}, \dotsc, a_{ic}\}$ to agent $A_i$ is within a factor of $c$ of the value of any intersecting sets in the optimal allocation.


Consider now the following instance of this problem, specified by the following agent types.
\begin{itemize}
\item For $1 \leq i \leq k-1$, agent $A_i$ desires $\{a_{i1},a_{i2},\dotsc,a_{ic}\}$ for value $k-i$ and $\{b_{i}\}$ for value $0$.
\item Agent $A_k$ desires $\{a_{k1},a_{k2},\dotsc,a_{kc}\}$ for value $k$ and $\{b_{k}\}$ for value $1$.
\item For $1 \leq i \leq k$, agents $B_i$ and $C_i$ both desire set $\{a_{i1}\}$ for value $(k-i)/c$.
\item For $1 \leq i \leq k$, agent $D_i$ desires set $\{a_{i1},a_{k1}\}$ for value $2(k-i)/c$.
\end{itemize}
Note that agent $A_k$ has a value density of $k/c$ for the desired set $\{a_{k1}, \dotsc, a_{kc}\}$, and each agent $A_i$ with $i < k$ has value density $(k-i)/c$ for desired set $\{a_{i1}, \dotsc, a_{ic}\}$.  Also, the agents $B_i$, $C_i$, and $D_i$ have a value density of $(k-i)/c$ for their desired sets.

We will suppose that $\alg$ applies the following fixed tie-breaking rules.  For any $i$, $\alg$ will break a tie between agents $A_i$, $B_i$, $C_i$, and/or $D_i$ first in favour of $D_i$, then in favour of $B_i$, then $A_i$.  We can also assume that $\alg$ breaks ties between multiple desired sets for agent $A_i$ in favour of $\{b_i\}$.  Finally, $\alg$ will favour allocating non-empty sets over allocating the empty set (e.g., if an agent declares the zero valuations).

We now describe a mixed Nash equilibrium for this problem instance.  Each agent $A_i$ declares the zero valuation.  Each agent $B_i$ and $C_i$ declares his valuation truthfully.  Each agent $D_i$ will declare his valuation truthfully with some probability $p_i$, and will otherwise declare the zero valuation.  We choose $p_i = \frac{1}{i+1}$.

What is the outcome when agents bid in this way?  First, each agent $A_i$ is allocated set $\{b_{i}\}$ (due to our assumed tie-breaking).  For the items $a_{ij}$, only items with $j=1$ will be allocated.  For $i < k$, if agents $D_1, \dotsc, D_{i-1}$ declare the zero allocation and $D_i$ does not, then object $a_{1i}$ will be allocated to $D_i$. If not, then item $a_{1i}$ will be allocated to agent $B_i$.  Item $a_{k1}$ will be allocated to $D_i$ where $i$ is the smallest such that $D_i$ does not declare the zero valuation, or $B_k$ if $D_1, \dotsc, D_k$ all declare the zero valuation.


We now argue that this distribution of declarations is indeed a mixed Nash equilibrium.  With probability 1, no agent $B_i$, $C_i$, or $D_i$ can obtain positive utility from any declaration (since their desired sets conflict with other bids of the same value density), so their distributions over declarations that obtain utility $0$ are necessarily optimal.  Furthermore, for each $i < k$, agent $A_i$ cannot obtain positive utility so his bidding strategy is also optimal.  Agent $A_k$ obtains utility $1$; his only hope for obtaining more utility is to declare a value less than $k-1$ for set $\{a_{k1},\dotsc,a_{kc}\}$.  However, if he declares some value $k-z$ with $z > 1$, say with $x = \lceil z \rceil$, then he can win his desired set only if bidders $D_1, \dotsc, D_{x-1}$ all bid the zero valuation, since otherwise an agent $D_j$ with $j < x$ would win his desired set, blocking the bid by agent $A_k$.  The probability that bidders $D_1, \dotsc, D_{x-1}$ all declare the zero valuation is $\frac{1}{2}\frac{2}{3}\dotsm\frac{x-1}{x} = \frac{1}{x} \leq \frac{1}{z}$. Thus, for any $z$, agent $A_k$ can obtain utility $z$ with probability at most $1/z$, for an expected utility of at most $1$.  The given declaration by agent $A_k$ is therefore optimal.

We will now bound the social efficiency of this equilibrium.  The optimal obtainable welfare is $k+\sum_{i=1}^{k-1} (k-i) = \frac{1}{2}k(k+1)$, by allocating set $\{a_{i1}, \dotsc, a_{ic}\}$ to agent $A_i$ for all $i$.  In the equilibrium we've described, object $b_k$ is allocated to agent $A_k$ for a value of $1$ and each object $a_{i1}$ for $i < k$ is allocated to either $B_i$ or $D_i$ at a per-item value of $(k-i)/c$.  For each $i<k$, object $a_{1k}$ will be allocated to bidder $D_i$ precisely if bidders $D_1, \dotsc, D_{j-1}$ declare the zero valuation but $D_i$ does not, which occurs with probability $\frac{1}{i(i+1)}$. Object $a_{1k}$ will be allocated to either $B_k$ or $D_k$ with the remaining probability, which is $\frac{1}{k}$.  Noting that each of $B_i$ and $D_i$ has a per-item value of $(k-i)/c$ for their desired sets, we conclude that the expected total value obtained is
\begin{align*} 
& 1 + \sum_{i < k}\frac{k - i}{c} + \sum_{i < k}\frac{1}{i(i+1)} \cdot \frac{k - i}{c} + \frac{1}{k} \cdot \frac{k - k}{c} \\
= & 1 + \frac{1}{c}\left[ \frac{1}{2}(k^2 - k) + k - \sum_{i < k}\frac{1}{i+1} - 1 \right ] \\
= & 1 + \frac{1}{c}\left[ \frac{1}{2}(k^2 + k) - H_k \right ]
\end{align*}
where $H_k$ is the $k$th harmonic number.

We conclude that the mixed price of anarchy for this mechanism is at least
\begin{align*}
\frac{\frac{1}{2}(k^2 + k)}{1 + \frac{1}{c}\left[\frac{1}{2}(k^2 + k) - H_k\right]} > c\left(\frac{k^2 + k}{k^2 + k + 2c - 2\ln k}\right)
\end{align*}
where we used the fact that $H_k > \ln k$.  Choose $k = e^{2c}$.  Then our mechanism has mixed price of anarchy at least
\begin{align*}
c\left(\frac{e^{4c} + e^{2c}}{e^{4c} + e^{2c} - 2c}\right) > c\left(\frac{e^{4c}}{e^{4c}-c}\right) > c\left(1 + \frac{c}{e^{4c}}\right)
\end{align*}
as required.
\end{proof}

\subsection{Correlated Types}

Recall that our bound for Bayesian Price of Anarchy 
required that agent types be distributed independently.  We now provide an alternative (weaker) bound that holds even if agent types are arbitrarily correlated.  The key to the new analysis is in considering a deviating behaviour for each agent that does not depend on the other agents' types.  The particular deviation we will consider is that of bidding half of one's true value for every set.  Our analysis will additionally require that the underlying allocation algorithm is a fixed order greedy algorithm.
\begin{thm}
\label{thm.bpoa.correlated}
Suppose $\alg$ is a $c$-approximate non-adaptive greedy algorithm for a combinatorial allocation problem.  Then $\mech_1(\alg)$ has Correlated Bayesian Price of Anarchy at most $4c$, for any type distribution $\dists$.
\end{thm}

The key to this result lies in the following lemma.  
\begin{lem}
\label{lem.greedy.smooth}
Suppose $\alg$ is a $c$-approximate non-adaptive greedy algorithm for a combinatorial allocation problem.  Then for all type profiles $\types$ and all strategy profiles $\bids(\cdot)$,
\[ \sum_i \utili( \vali/2, \bidsmi(\valsmi) ) \geq \frac{1}{2c} SW_{OPT}(\vals) - SW(\alg(\bids(\vals)), \vals). \]
\end{lem}
\begin{proof}
Let $\mathbf{y}$ denote the optimal allocation for type profile $\types$.  
Choose agent $i$, and consider the outcome of $\alg$ on input profile $(\vali/2, \bidsmi(\valsmi))$.  Let $\alloci = \alg_i(\vali/2, \bidsmi(\valsmi))$.  Note that it must either be that $\criti(y_i, \bidsmi(\valsmi)) \geq \frac{1}{2}\vali(y_i)$ or not.  In the latter case, agent $i$ must obtain some allocation $\alloci$ with $\rank(i, \alloci, \vali(\alloci)/2) \geq \rank(i, y_i, \vali(\alloci)/2)$.  Since $\alg$ is a non-adaptive greedy algorithm, this then implies that $\vali(\alloci) \geq \frac{1}{c}\vali(y_i)$, since otherwise $\alg$ would obtain less than a $\frac{1}{c}$ fraction of the optimal social welfare on the input in which agent $i$ places bids only on sets $\alloci$ and $y_i$, and all other agents bid $0$.

We conclude that for all $i$, either $\criti(y_i, \bidsmi(\valsmi)) > \frac{1}{2}\vali(y_i)$ or else $\vali(\alloci) \geq \frac{1}{c}\vali(y_i)$.  Let $N = \{ i \ |\ \criti(y_i, \bidsmi(\valsmi)) > \frac{1}{2}\vali(y_i)\}$ be the set of agents for which the former condition holds.  We then note that
\begin{align*}
\sum_{i \in N} \frac{1}{2} \vali(y_i) < \sum_{i \in N} \criti(y_i, \bidsmi(\valsmi)) \leq c SW( \alg(\bids(\vals)), \bids(\vals)) \leq c SW( \alg(\bids(\vals)), \vals )
\end{align*}
where the second inequality is due to Lemma \ref{lem.greedy} and the third is due to Lemma \ref{lem.firstprice.overbid}.  Furthermore, since $\vali(\alloci) \geq \frac{1}{c}\vali(y_i)$ for all $i \not\in N$, we have
\begin{align*}
\sum_{i \not\in N} \frac{1}{2} \vali(y_i) \leq \sum_{i \not\in N} \frac{c}{2} \vali(\alloci(\vali/2, \bidsmi(\valsmi))) \leq c \sum_i \utili( \vali/2, \bidsmi(\valsmi))
\end{align*}
where the second inequality follows because we are using the first-price payment scheme.  Combining these inequalities yields
\[
\sum_i \utili( \vali/2, \bidsmi(\valsmi) ) + SW(\alg(\bids(\vals)), \vals) \geq \frac{1}{2c} SW_{OPT}(\vals)
\]
as required.
\end{proof}

Theorem \ref{thm.bpoa.correlated} now follows easily from Lemma \ref{lem.greedy.smooth}.  Recall that Lemma \ref{lem.greedy.smooth} holds for all strategy profiles, not just strategies in equlibrium.  If we take $\bids$ to be an equilibrium profile under type distribution $\dists$, then
\begin{align*}
E_{\types}[SW(\alg(\bids(\vals)), \vals)] & \geq \E_{\types}\left[\sum_i \utili(\bids(\types))\right] \\
& = \sum_i \E_{\typei} \E_{\typesmi | \typei}[\utili(\bidi(\typei), \bidsmi(\typesmi))] \\
& \geq \sum_i \E_{\typei} \E_{\typesmi | \typei}\left[\utili\left(\frac{\typei}{2}, \bidsmi(\typesmi)\right)\right] \\
& = E_{\types}\left[\sum_i \utili\left(\frac{\typei}{2}, \bidsmi(\typesmi)\right)\right] \\
& \geq E_{\types}\left[ \frac{1}{2c}\sum_i SW_{OPT}(\types) - SW( \alg(\bids(\types)), \types )\right] \quad \text{(Lemma \ref{lem.greedy.smooth})}
\end{align*}
from which we conclude that
\[ E_{\types}[SW(\alg(\bids(\vals)), \vals)] \geq \frac{1}{4c} E_{\types}\left[ \sum_i OPT(\types) \right] \]
completing the proof of Theorem \ref{thm.bpoa.correlated}.

\section{Critical-Price Mechanisms}
\label{sec:ex-post}

We begin by studying the performance of critical (i.e.\ second) price mechanisms at equilibrium.  The mechanism we study is $\mech_{2}(\alg)$, which is defined with respect to an arbitrary monotone strongly loser-independent algorithm $\alg$.  Recall that $\mech_{2}(\alg)$ proceeds by first collecting a declaration profile from the agents, then passing the observed declarations to $\alg$ as input.  The mechanism returns the allocation provided by $\alg$ as output, and charges each agent his critical value for the set received (computed via additional calls to $\alg$; see Section \ref{sec:impl-crit}).

We will show that every Bayes-Nash equilibria of $\mech_{2}(\alg)$ has a social welfare guarantee nearly matching that of the original algorithm $\alg$.  This result requires that we make an assumption on the bidding strategies applied by the users; namely, that they do not \emph{overbid}, meaning that they do not bid more than their true value on any given set $S$.  This overbidding assumption is necessary to exclude certain degenerate equilibria, such as one agent making an infinitely large bid on the set of all objects and other bidders bidding $0$.  We note that such assumptions are reasonable in general; even the truthful Vickrey auction of a single item requires a no-overbidding assumption to bound the efficiency of the outcome at equilibrium.  In Section \ref{sec.overbid} we discuss ways to relax this assumption by modifying the mechanism slightly.

\subsection{Bayes-Nash Equilibria}

We 
begin by analyzing 
the Bayesian price of anarchy for the critical price mechanism $\mech_{2}(\alg)$.  Given that agents will not overbid, a simple modification of 
Theorem \ref{thm.mixedNE}
yields a result for BNE under critical prices.  

\begin{thm}
\label{thm.mixedNEcrit}
Suppose $\alg$ is a $c$-approximate 
monotone strongly loser independent allocation rule,
and that $\bids$ is a Bayes-Nash equilibrium of $\mech_{2}(\alg)$ in which agents do not overbid.  Then the expected welfare when agents declare according to $\bids$ is a $(c+1)$-approximation to the expected optimal welfare.
\end{thm}


\begin{lem}
\label{lem.bayes.crit}
Suppose that $\bids$ is a 
Bayes-Nash equilibrium for mechanism $\mech_{2}(\alg)$ and distribution $\dists$.  Then for all $i$, all $\typei$, and all $S \subseteq M$,
\[
   \E_{\typesmi}[\criti(S,\bidsmi(\typesmi))] \geq \typei(S) - \E_{\typesmi}[\typei(\alloci(\bidi(\typei),\bidsmi(\typesmi)))]
\]
\end{lem}
\begin{proof}
Choose any $i$, $\typei$, and $S$.  Let $\decli$ be a single-minded declaration for set $S$ at value $\typei(S)$, and consider a strategy under which agent $i$ declares $\decli$ when his type is $\typei$.  Under this strategy, the expected utility of agent $i$ with type $\typei$ is
\begin{equation}
\label{eq.crit.bne.1}
\begin{split}
  \E_{\typesmi}[\utili( \decli, \bidsmi(\typesmi) )] & \geq \E_{\typesmi}[ \max\{\typei(S) - \criti(S, \bidsmi(\typesmi)), 0\} ] \\
& \geq \typei(S) - \E_{\typesmi}[ \criti(S, \bidsmi(\typesmi)) ].
\end{split}
\end{equation}
Since $\bidi$ is 
an equilibrium strategy for agent $i$, it must be that
\begin{equation}
\label{eq.crit.bne.2}
\begin{split}
  \E_{\typesmi}[\utili( \decli, \bidsmi(\typesmi) )] & \leq \E_{\typesmi}[\utili( \bidi(\typei), \bidsmi(\typesmi) )] \\
& \leq \E_{\typesmi}[\typei(\alloci(\bidi(\typei),\bidsmi(\typesmi)))].
\end{split}
\end{equation}
Combining equations \eqref{eq.crit.bne.1} and \eqref{eq.crit.bne.2} leads to the desired result.
\end{proof}

Following the proof of Theorem \ref{thm.mixedNE}, we conclude that for all equilibria $\bids$, if we write $\mathbf{y}^{\types}$ for an optimal allocation for any given type profile $\types$, then
\begin{equation}
\label{eq.crit.poa.1}
\begin{split}
\E_{\types}\left[\sum_i\E_{\typesmi'}[\criti(y_i^\types,\bidsmi(\typesmi'))]\right]
& \geq \E_{\types}\left[\sum_i\typei(y_i^\types)\right] \\
&  \ \ \ - \E_{\types}\left[\sum_i \E_{\typesmi'}[\typei(\alloci(\bidi(\typei),\bidsmi(\typesmi')))]\right].
\end{split}
\end{equation}
Just as in the proof of Theorem \ref{thm.mixedNE}, we obtain the bounds
\[
\E_{\types}\left[\sum_i\typei(y_i^\types)\right] = \E_{\types}[SW_{OPT}(\types)],
\]
\[
\E_{\types}\left[\sum_i \E_{\typesmi'}[\typei(\alloci(\bidi(\typei),\bidsmi(\typesmi')))]\right] = \E_{\types}[SW(\alg(\bids(\types)),\types)],
\]
\[
\E_{\types}\left[\sum_i\E_{\typesmi'}[\criti(y_i^\types,\bidsmi(\typesmi'))]\right] \leq c \E_{\types}[ SW(\alg(\bids(\types)),\types) ]
\]
which, taken together with \eqref{eq.crit.poa.1}, completes the proof of Theorem \ref{thm.mixedNEcrit}.  Note that when deriving the last inequality above, we do not invoke Lemma \ref{lem.firstprice.overbid} (as in the proof of Theorem \ref{thm.mixedNE}); instead, we use the assumption that agents do not overbid.
 \QED

In precisely the same way as for the first-price mechanism, the bound on the price of anarchy also extends to coarse correlated equilibrium.

\begin{cor}{(of proof)}
\label{cor:coarse.crit}
The bound of $(c+1)$ on the price of anarchy applies also to coarse correlated equilibrium. 
\end{cor}

We next show that this gap between the approximation factor of the original algorithm and the price of anarchy of the critical price mechanism is required for large $c$.  
For any $c \geq 1$ we exhibit a combinatorial allocation problem and a non-adaptive greedy algorithm $\alg$ such that the approximation factor of $\alg$ is $c + \frac{1}{c}$ but the (pure) price of anarchy of $\mech_{2}(\alg)$ is $c+1$.
This leads us to conclude that, in general, the bound in Theorem \ref{thm.mixedNEcrit} cannot be improved beyond 
$c + 1 - O(\frac{1}{c})$.

\begin{prop}
For any $c \geq 1$, there is a combinatorial allocation problem $\calp$ and a non-adaptive greedy algorithm $\alg$ such that $\alg$ is a $(c+\frac{1}{c})$-approximation for $\calp$, and the pure price of anarchy for $\mech_{2}(\alg)$ is $c + 1$.
\end{prop}
\begin{proof}
Consider a combinatorial auction problem with two objects $a,b$ and two players, under the restriction that each player can be allocated at most one object and player $2$ cannot be allocated object $b$.
Algorithm $\alg$ will be a non-adaptive greedy algorithm defined by the
following ranking function: $r(1,\{a\},v) = v$, $r(1,\{b\},v) = c \cdot v$, and $r(2,\{a\},v) = \frac{1}{c} \cdot v$ where ties are broken in favor of $r(1,\{a\},v)$.
As a result of this ranking function and greediness, it follows that 
if $\vali[1](a) \geq c \vali[1](b)$ and $\vali[1](a) \geq \frac{1}{c}\vali[2](a)$ then $a$ is allocated to player $1$ and $\emptyset$ to player $2$; otherwise 
$b$ is allocated to player $1$ and $a$ to player $2$.  

Note that this is a $c+\frac{1}{c}$ approximation algorithm, since whenever the algorithm allocates $a$ to player $1$ we have $\vali[2](a) + \vali[1](b) \leq (c+\frac{1}{c})\vali[1](a)$, and whenever the algorithm allocates $a$ to player $2$ we have $\vali[1](a) \leq c( \vali[1](b) + \vali[2](a) )$. In either case, 
the algorithm's allocation has social welfare that is at least a
$(c+\frac{1}{c})$ fraction of
the other alternative. 

Consider the mechanism $\mech_{2}(\alg)$, and suppose that the agents have a type profile in which $\vali[1](a) = \vali[1](b) = 1$ and $\vali[2](a) = c$.  Then the declaration profile $\decli[1](a) = 1$, $\decli[1](b) = 0$, and $\decli[2](a) = 0$ is in equilibrium, since agent $1$ cannot improve upon his utility of $1$ and agent $2$ cannot affect the outcome without paying at least $\criti[2](a,\decli[1]) = c$, for a utility of $0$.  The social welfare at this equilibrium is $1$, but a total of $c+1$ is possible by allocating $a$ to player $2$ and allocating $b$ to player $1$.  Thus the price of anarchy for $\mech_{2}(\alg)$ is at least $c+1$.
\end{proof}


\subsection{Correlated Types}
\label{sec.crit.correlated}

Theorem \ref{thm.mixedNEcrit} requires that agent types be distributed independently.
As with the first-price mechanism, we can provide a somewhat weaker bound 
that holds even when agent types are arbitrarily correlated.  
And as in Theorem \ref{thm.mixedNEcrit} this result additionally requires that the underlying allocation algorithm is a non-adaptive greedy algorithm.
\begin{thm}
\label{thm.bpoa.crit.correlated}
Suppose $\alg$ is a $c$-approximate non-adaptive greedy algorithm for a combinatorial allocation problem, and that agents do not overbid.  Then $\mech_{2}(\alg)$ has Correlated Bayesian Price of Anarchy at most $4c$, for any type distribution $\dists$.
\end{thm}

The proof of Theorem \ref{thm.bpoa.crit.correlated} follows that of Theorem \ref{thm.bpoa.correlated} almost exactly.  The sole difference is that the invocation of Lemma \ref{lem.firstprice.overbid} in the proof of Theorem \ref{thm.bpoa.correlated} is replaced by an appeal to the no-overbidding assumption.  We omit the details for brevity.

\subsection{Overbidding and Restricted Expressiveness}
\label{sec.overbid}

Our analysis to this point made use of a no-overbidding assumption, which states that no agent will place a bid larger than her true value on any given set.  
However, our use of the no-overbidding assumption is marred by the fact that a restriction to no-overbidding strategies is not always rational when agents have complete confidence about their opponents' type distributions.  As the following example shows, an agent may be \emph{strictly} better off by overbidding, even in a full-information setting.  In other words, a strategy with overbidding is not necessarily dominated.


\begin{example}
\label{ex.overbid}
Consider a combinatorial auction with 3 objects, $\{a,b,c\}$, and 3 bidders, under the feasibility restriction that each agent can be allocated at most one object.  Let $\alg$ be the greedy algorithm that orders bids by value.  Suppose the types of the players are as follows: $t_1(b) = 2$, $t_1(c) = 4$, $t_2(c) = 3$, $t_3(a) = 1$, $t_3(b) = 6$, and all other values are $0$.  Consider the following bidding strategies for agents $2$ and $3$: bidder $2$ declares truthfully with probability 1, and bidder $3$ either declares  single-mindedly for $a$ with value $1$, or single-mindedly for $b$ with value $6$, each with equal probability.

How should agent $1$ declare to maximize utility?  We can limit our analysis to pure strategies (as any optimal randomized strategy has only optimal strategies in its support).  Suppose agent $1$ does not overbid and declares at most $2$ for object $b$.  If he also declares at least $3$ for object $c$, then he wins $c$ with probability $1$ for an expected utility of $1$.  If he doesn't declare at least $3$ for object $c$, then he wins $b$ with probability $1/2$ and nothing otherwise, again for an expected utility of $1$.  So agent $1$ can gain a utility of at most $1$ if he does not overbid.  If, however, he declares $5$ for $b$ and $4$ for $c$, then he wins $b$ with probability $1/2$ and wins $c$ otherwise, for an expected utility of $3/2$.  If agent $1$ bids in this way, the resulting combination of strategies forms a mixed Nash equilibrium.  Thus, in mixed equilibria, an agent may strictly improve his utility by overbidding.
\end{example}

We now show that if we modify mechanism $\mech_{2}(\alg)$ by effectively limiting the expressiveness of the bids made by the agents, then we obtain the same efficiency bounds at equilibria but furthermore guarantee that any bidding strategy that involves overbidding is dominated.  Thus, as long as agents avoid dominated strategies (a very mild assumption), all equilibria of rational play lead to approximately efficient outcomes.  

For a $\sli$ $\alg$, the modified mechanism $\mech_{2}^{*}(\alg)$
is as described in Figure \ref{fig.mech.regret}.
Mechanism $\mech_2^*(\alg)$ proceeds by first simplifying the declaration given by each agent, then passing the simplified declarations to algorithm $\alg$.  The resulting allocation is paired with a payment scheme that charges critical prices.  

\begin{figure}
\begin{center}
	\fbox{
		\begin{minipage}{0.95\linewidth}

\textbf{Mechanism} $\mech_{2}^{*}(\alg)$:
\medskip
\hrule

\medskip

\textbf{Input:} Declaration profile $\decls = \decl_1, \dotsc, \decl_n$.

%
\begin{tabbing}
1. \= $\decl' \leftarrow $ \texttt{SIMPLIFY}$(\decls)$. \\
2. \> Allocate $\alg(\decls')$, charge critical prices.
\end{tabbing}
		\end{minipage}
	}
	\fbox{
		\begin{minipage}{0.95\linewidth}

\textbf{Procedure} \texttt{SIMPLIFY:}
\medskip
\hrule

\medskip
\textbf{Input:} Declaration profile $\decls = \decl_1, \dotsc, \decl_n$.

%
\begin{tabbing}
1. \= For each $i \in [n]$: \\
2. \> \quad \= Choose $S_i \in \argmax_S\{\decli(S)\}$, breaking \\ 
\>\> ties in favour of smaller sets. \\
3. \> \> $\decli' \leftarrow (S_i,\decli(S_i))$. \\
4. \> Return $(\decl_1', \dotsc, \decl_n')$.
\end{tabbing}
		\end{minipage}
	}
   \caption{Simplifying declarations in a critical price mechanism.} 
			\label{fig.mech.regret}
	\end{center}
\end{figure}


The simplification process \texttt{SIMPLIFY} essentially converts any declaration into a single-minded declaration (and does not affect declarations that are already single-minded).  We can therefore assume without loss of generality that agents always make single-minded declarations to this mechanism, as additional information is not used.\footnote{We note, however, that this is not the same as assuming that agents are single-minded; our results hold for bidders with general private valuations.} 

Fix a particular combinatorial auction problem and type profile $\types$, and let $\alg$ be an arbitrary strongly loser-independent approximation algorithm.  Since $\vals$ is fixed, we can think of a strategy for each agent $i$ as a declaration $\decli \in V_i$.  Let $\decls$ be a declaration profile; we suppose each $\decli$ is a single-minded bid for set $S_i$ (and, in general, we will write $S_i$ for the desired set in declaration $\decli$).
We draw the following conclusion about the bidding choices of rational agents.

\begin{lem}
\label{lem.undom1}
Let $\alg$ be \sli, and fix type profile $\types$.
Then for each agent $i$, a single-minded declaration $\decli$ for set $S_i$ is an undominated strategy for mechanism $\mech_{2}^{*}(\alg)$ if and only if $\decli(S_i) = \typei(S_i)$.
\end{lem}
\begin{proof}
Fix some $\declsmi$ and suppose $\decli$ is a single-minded declaration for set $S_i$ 
On input $(\decli,\declsmi)$, mechanism $\mech_{2}^{*}(\alg)$ either allocates $S_i$ or $\emptyset$ to agent $i$.  
Thus agent $i$'s utility for declaring $\decli$, $\utili(\decli,\declsmi)$, is $\typei(S_i) - \crit_i(S_i,\declsmi)$ when $\decli(S_i) > \crit_i(S_i,\declsmi)$, and $0$ otherwise (where $\criti$ denotes critical prices with respect to $\mech_{2}^{*}(\alg)$).  A declaration of $\decli(S_i) = \typei(S)$ therefore maximizes $\utili(\decli,\declsmi)$ for all $\declsmi$.

Next suppose that $\decli(S_i) \neq \typei(S_i)$; we will show that $\decli$ is 
dominated. Let $\decli'$ be the single-minded declaration for $S_i$ at value $\typei(S_i)$.  Suppose there is some $\declsmi$ such that $\criti^\alg(S_i,\declsmi)$ lies strictly between $\decli(S_i)$ and $\typei(S_i)$.  For simplicity we will assume such a $\declsmi$ exists; handling the general case requires only a technical extension of notation\footnote{If $\criti^\alg(S_i,\declsmi)$ never lies between $\decli(S_i)$ and $\typei(S_i))$ for any $\declsmi$, then $\mech_\alg(\decli,\declsmi) = \mech_\alg(\decli',\declsmi)$ for all $\declsmi$, so $\decli$ and $\decli'$ are equivalent strategies.  We can therefore think of $\decli$ as being ``the same'' as a single-minded declaration for $S_i$ at value $\typei(S_i)$.  We will ignore this technical issue for the remainder of the proof, in the interest of clarity.}.  Then if $\decli(S_i) < \typei(S_i)$, then $\utili(\decli',\declsmi) > 0 = \utili(\decli,\declsmi)$.  Otherwise, if $\decli(S_i) > \typei(S_i)$), then $\utili(\decli',\declsmi) \geq 0 > \utili(\decli,\declsmi)$.  Thus, in either case, we have $\utili(\decli',\declsmi) > \utili(\decli,\declsmi)$, and therefore
declaration $\decli'$ strictly dominates declaration $\decli$.  
\end{proof}

Given Lemma \ref{lem.undom1}, we can analyze the efficiency of equilibria of $\mech_{2}^{*}(\alg)$ in a manner identical to $\mech_{2}(\alg)$.  Rather than explicitly assuming that agents do not overbid, Lemma \ref{lem.undom1} implies that they will not.

\begin{thm}
\label{thm.mixedNEcritstar}
Suppose $\alg$ is a $c$-approximate 
monotone strongly loser independent allocation rule,
and that $\bids(\cdot)$ is a Bayes-Nash equilibrium of $\mech_{2}^{*}(\alg)$.  Then the expected welfare when agents declare according to $\bids$ is a $(c+1)$ approximation to the expected optimal welfare.
\end{thm}



\subsection{Calculating Critical Prices}
\label{sec:impl-crit}

For many allocation algorithms (such as all of the algorithms discussed in Section \ref{appendix:applications}), the calculation of critical prices is a simple task, which can be performed in parallel with the computation of an allocation profile.  We leave the development of such pricing methods to the creators of the allocation algorithms to which our reduction may be applied.  However, even if a specially-tailored algorithm for computing exact critical prices is not available, we note that critical prices for a given black-box greedy algorithm can be determined to within an additive $\eps$ error in polynomial time via simple binary search.  Thus, assuming that valuation space is discretized by multiples of $\epsilon$, critical prices can be determined efficiently.  If valuation space is continuous, then our interpretation is that any equilibrium for the (exact) critical-price mechanism will be an (additive) $\epsilon$-approximate equilibrium for a mechanism that uses $\epsilon$-approximate critical prices.

We now describe the procedure for determining critical prices in more detail. Fix greedy allocation rule $\alg$, agent $i$, and declarations $\decls$.  Suppose that $\algi(\decli,\declsmi) = S$.  
We wish to resolve the value of $\criti(S,\declsmi)$ in the range $[0, \decli(S)]$ using binary search in the following way.  For all $z \geq 0$, write $\decli^z$ for the single-minded declaration for set $S$ at value $z$. Given query value $z \in [0,\decli(S)]$, we check if $\alg_i(\decli^z, \decli) = S$.  If so, decrease the value of $z$; otherwise, increase the value of $z$.  Since $\alg$ is monotone, we have that $\alg_i(\decli^z,\declsmi) = S$ if and only if $z > \criti(S,\declsmi)$.
This procedure resolves the value of $v$ to within $\eps$ in $O(\log \decli(S)/\eps)$ iterations.  Thus, for any given input to mechanism $\mech_{2}(\alg)$, the critical prices for all agents' allocated sets can be found in $O(n\log(v_{max}/\eps))$ invocations of algorithm $\alg$, where $v_{max} = \max_{i,S}\decli(S)$.

\section{Conclusion and open problems}
\label{sec:conclusions}


A central theme in algorithmic mechanism design concerns the transformation of 
algorithms into mechanisms that satisfy some game-theoretic solution concept (e..g incentive compatability, approximations at equilbrium). In contrast to incentive
compatibility (where generally we do not expect to be able to preserve
approximation bounds), we show that for a wide class of greedy
algorithms, approximation bounds for combinatorial allocation algorithms can be transformed into mechanisms
that enjoy closely matching price of anarchy bounds. Notably, these results
apply to Bayesian equilibira and some forms of repeated auctions.

We leave open a number of interesting challenges. 
Our results are motivated by, and pertain to, monotone greedy algorithms as formally defined 
in Section~\ref{sec:greedy}. In fact, the key property of such algorithms
are that they are monotone strongly loser-independent as defined in
Section~\ref{sec:strong-loser-indep} and, with the exception of 
the results for correlated equilibria and best response dynamics, 
our results hold for arbitrary monotone strongly loser-independent algorithms. 
In particular, our result for correlated equilibria of the first-price mechanism
requires that the allocation algorithm $\alg$ is a fixed order greedy algorithm 
and achieves a price of anarchy bound of $4c$, in contrast 
to our $c + o(1)$ result for independent agent distributions. Can the
price of anarchy bound for correlated equilibria be improved?  Can it
be extended to adaptive greedy or more generally strongly loser
independent algorithms? 

Greedy algorithms for allocation problems often provide the best 
known approximations for combinatorial auction problems, but are 
nevertheless a restricted class of algorithms. 
The basic open question in this regard is: for what class of allocation algorithms 
can a given approximation algorthm $\alg$ 
be transformed into a deterministic or randomized mechanism $\calm(\alg)$ 
that provides a POA bound 
(closely) matching $\alg$'s approximation ratio? 
We also note that our framework does not capture 
\emph{all} algorithms that are typically thought of as greedy, since our
definition assumes that it is the player-allocation pairs that are considerd
greedily.  This excludes, for example, the greedy algorithm 
for combinatorial auctions where the valuation function of every agent is a
monotone submodular function. That algorithm  
considers each item (in any arbitrary order) and awards it to the agent 
having the maximum marginal gain for that item. 
This suggests the question as to whether or not 
price of anarchy results could be extended to more general forms of greedy
allocation rules.
Similarly, the recent
Buchbinder et al \cite{Buchbinder} randomized online algorithm for unconstrained 
non-monotone submodular maximzation also considers items (rather than bids) 
in a greedy alglorithm. 
Can our methodology be extended to include non monotone combinatorial
auctions (i.e. no free disposal)?  It is also interesting to consider more general
settings of incomplete information, such as interdependent valuations ---
See Roughgarden and Talgam-Cohen \cite{RoughgardenT13}. 
%

\bibliographystyle{plain}
\bibliography{thesis}

\appendix

\section{Existence of Pure Nash Equilibria}
\label{sec:pureNE}

As stated in section \ref{sec:mixedBNE}, the power of our pure price of anarchy bounds such as Theorem \ref{thm.pureNE}, is marred by the fact that, 
for some problem instances, the mechanism $\mech_1(\alg)$ is not 
guaranteed to have a pure Nash equilibrium. 
This is true even under the assumption that private valuations and payments are discretized, so that all values and payments are multiples of some aribtrarily small $\eps > 0$.
A simple example for $\mech_1(\alg)$ is given below.  

\begin{example}
\label{example:no-pure}
Consider an instance of the combinatorial auction problem with two objects, $M = \{a,b\}$, and three agents.  Our feasibility constraint is that each agent can be assigned at most one object, and moreover agent $2$ cannot be allocated object $b$ and agent $3$ cannot be allocated object $a$.  Let $\alg$ be the greedy algorithm that ranks bids by value.  Suppose the true types of the agents are as follows: $\vali[1](a) = 4$, $\vali[1](b) = 2$, $\vali[2](a) = 3$, $\vali[2](b) = 0$, $\vali[3](a) = 0$, and $\vali[3](b) = 3$.  

We now prove that no pure Nash equilibrium exists for this example, even if we assume that agents declare multiples\footnote{That is, our lack of pure equilibrium is not due to the possibility of infinitesimal improvements.  One can also interpret our example as demonstrating that there is no $(1+\epsilon)$-approximate pure Nash equilibrium for small $\epsilon > 0$.} of some $\eps > 0$.  Assume for contradiction that there is a Nash equilibrium $\decls$ for type profile $\types$ and mechanism $\calm_1(\alg)$.

We know that agent $1$ does not win item $b$ with a payment greater than $2$, as this would cause him negative utility (so he would certainly not be in equilibrium).  Thus it must be that $\algi[3](\decls) = \{b\}$, since otherwise agent 3 could change his declaration to win $\{b\}$ and increase his utility.  Thus, since agent 1 does not win item $\{b\}$, we conclude that $\algi[1](\decls) = \{a\}$, since otherwise agent 1 could change his declaration to win $\{a\}$ and increase his utility.

Now note that if $\decli[1](\{a\}) < 3$, agent 2 could increase his utility by making a winning declaration for $\{a\}$.  Thus $\decli[1](\{a\}) \geq 3$, and hence $\utili[1](\decls) \leq 4-3 = 1$.  This also implies that $\decli[1](\{a\}) > \decli[1](\{b\})$, so agent 3 would win $\{b\}$ regardless of his bid.  Thus, since agent 3 maximizes his utility up to an additive $\eps$, it must be that $\decli[3](\{b\}) \leq \eps$.  But then agent 1 could improve his utility by changing his declaration and bidding 0 for $\{a\}$ and $2\eps$ for $\{b\}$, obtaining utility $2 - 2\eps > 1$.  Therefore $\decls$ is not an equilibrium, a contradiction.
\end{example}


\section{Combining Mechanisms}
\label{sec.combine}



%

A standard technique in the design of allocation rules is to consider both a greedy rule that favours allocation of small sets, and a simple rule that allocates all objects to a single bidder, and apply whichever solution obtains the better result \cite{BGN-03,BKV-05,Halldorsson00,MN-08}.  When bidders are single-minded, such a combination rule will be incentive-compatible \cite{MN-08}.  We would like to extend our results to cover rules of this form, but 
the price of anarchy for such a rule (with either the first-price or critical-price payment scheme) may be much worse than its combinatorial approximation ratio.  Consider the following example.

\begin{example}
\label{ex.combine}
Consider the combinatorial auction problem.  Suppose $\alg$ is the non-adaptive greedy algorithm with priority rule $r(i,S,v) = v$ if $|S| \leq \sqrt{m}$, and $r(i,S,v) = 0$ otherwise.  Let $\alg'$ be the non-adaptive greedy algorithm with priority rule $r(i,S,v) = v$ if $S = M$, and 
$r(i,S,v) = 0$ otherwise.  Then $\alg'$ simply allocates the set of all objects to the player that declares the highest value for it.  Let $\alg_{max}$ be the allocation rule that applies whichever of $\alg$ or $\alg'$ obtains the better result; that is, on input $\decls$, $\alg_{max}$ returns $\alg(\decls)$ if $SW(\alg(\decls),\decls) > SW(\alg'(\decls),\decls)$, otherwise returns $\alg'(\decls)$.  It is known that $\alg_{max}$ is a $O(\sqrt{m})$ approximate algorithm \cite{MN-08}.

Our instance of the CA problem is the following.  We have $n = m \geq 2$, say with $M = \{a_1, \dotsc, a_m\}$.  Choose $\eps > 0$ arbitrarily small.  For each $i$, the private type of agent $i$, $\typei$, is the pointwise maximum of two single-minded valuation functions: one for set $\{a_i\}$ at value $1$, and the other for set $M$ at value $1 + \eps$.  An optimal allocation profile for $\types$ would assign $\{a_i\}$ to each agent $i$, for a total welfare of $m$.  

We construct a declaration profile as follows.  For each $i$, $\decli$ is the single-minded valuation function for set $M$ at value $1 + \eps$.  On input $\decls$, $\alg_{max}$ will assign $M$ to some agent, for a total welfare of $1+\eps$.  Also, $\decls$ is a pure Nash equilibrium for $\calm_1(\alg_{max})$ and $\calm_{crit}(\alg_{max})$: all agents receive a utility of $0$, and there is no way for any single agent to obtain positive utility by deviating from $\decls$.  Taking $\eps \to 0$, we conclude that the price of anarchy for any of these mechanisms is $\Om(m)$, which does not match the combinatorial $O(\sqrt{m})$ approximation ratio of $\alg_{max}$.
\end{example}


In light of the example above, one must consider different ways to combine two allocation rules.  For instance, one could implement each rule as a separate mechanism, then randomly choose between the two with equal probability.  
Such an approach can work well when the two allocation rules work with disjoint parts of the declaration space, so that agents can optimize their bids separately for each mechanism.



\end{document}